\documentclass[english,Journal]{IEEEtran}
\usepackage[T1]{fontenc}
\usepackage[latin9]{inputenc}
\usepackage{textcomp}
\usepackage{amsthm}
\usepackage{amsmath}
\usepackage{graphicx}
\usepackage{amssymb}
\usepackage{subfig}
\usepackage{cite}

\newcommand{\lyxmathsym}[1]{\ifmmode\begingroup\def\b@ld{bold}
  \text{\ifx\math@version\b@ld\bfseries\fi#1}\endgroup\else#1\fi}

\theoremstyle{plain}
\newtheorem{thm}{Theorem}
  \theoremstyle{remark}
  \newtheorem{rem}{Remark}
  \theoremstyle{plain}
  \newtheorem{lem}{Lemma}

\begin{document}

\title{Minimum Communication Cost for Joint Distributed Source Coding and Dispersive Information Routing}

\author{Kumar Viswanatha, ~\IEEEmembership{Student Member,~IEEE,} Emrah Akyol, ~\IEEEmembership{Student Member,~IEEE}\\ and Kenneth Rose, ~\IEEEmembership{Fellow,~IEEE}
\thanks{Authors are with the Department of Electrical and Computer Engineering, University of California, Santa Barbara, CA, 93106 USA e-mail: {kumar, eakyol, rose} @ece.ucsb.edu}\\
\thanks{The work was supported by the NSF under grants CCF-0728986, CCF - 1016861 and CCF-1118075. The material in this paper was presented in part at the IEEE international symposium on information theory at Austin, USA, Jun 2010, IEEE information theory workshop at Dublin, Ireland, Aug 2010  and the IEEE international symposium on information theory at Saint Petersburg, Russia, Aug 2011}}

%\markboth{IEEE TRANSACTIONS ON INFORMATION THEORY, VOL. XX, NO. XX, XX XX}
%{IEEE TRANSACTIONS ON INFORMATION THEORY, VOL. XX, NO. XX, XX XX}

\maketitle
\begin{abstract}
This paper considers the problem of minimum cost communication of
correlated sources over a network with multiple sinks, which consists of distributed source coding followed by routing. We introduce a new routing paradigm called dispersive information
routing, wherein the intermediate nodes are allowed to `split' a packet
and forward subsets of the received bits on each of the forward paths.
This paradigm opens up a rich class of research problems which focus
on the interplay between encoding and routing in a network. Unlike
conventional routing methods such as in \cite{Networked_Slepian_Wolf}, dispersive
information routing ensures that each sink receives just the information 
needed to reconstruct the sources it is required to reproduce. We demonstrate
using simple examples that our approach offers better asymptotic performance
than conventional routing techniques. This paradigm leads to a new
information theoretic setup, which has not been studied earlier. We
propose a new coding scheme, using principles from multiple descriptions
encoding \cite{VKG} and Han and Kobayashi decoding \cite{Han_Kobayashi}.
We show that this coding scheme achieves the complete rate region
for certain special cases of the general setup and thereby achieves the
minimum communication cost under this routing paradigm. \end{abstract}

\section{Introduction\label{sec:Introduction}}

Compression of sources in conjunction with communication over a network
has been an important research area, notably with the recent advancements
in distributed compression of correlated sources and network (routing)
design, coupled with the deployment of various sensor networks. Encoding
correlated sources in a network, such as a sensor network with multiple
nodes and sinks as shown in Fig. \ref{fig:A-general-multi-sink},
has conventionally been approached from two different directions.
The first approach is routing the information from different sources
in such a way as to efficiently re-compress the data at intermediate
nodes without recourse to distributed source coding (DSC) methods
(we refer to this approach as joint coding via `explicit communication').
Such techniques tend to be wasteful at all but the last hops of the
communication path. The second approach performs DSC followed by simple
routing. Well designed DSC followed by optimal routing can provide
good performance gains. We will focus on the latter category.
Relevant background on DSC and route selection in a network is given
in the next section.

\begin{figure}
\centering\includegraphics[scale=0.35]{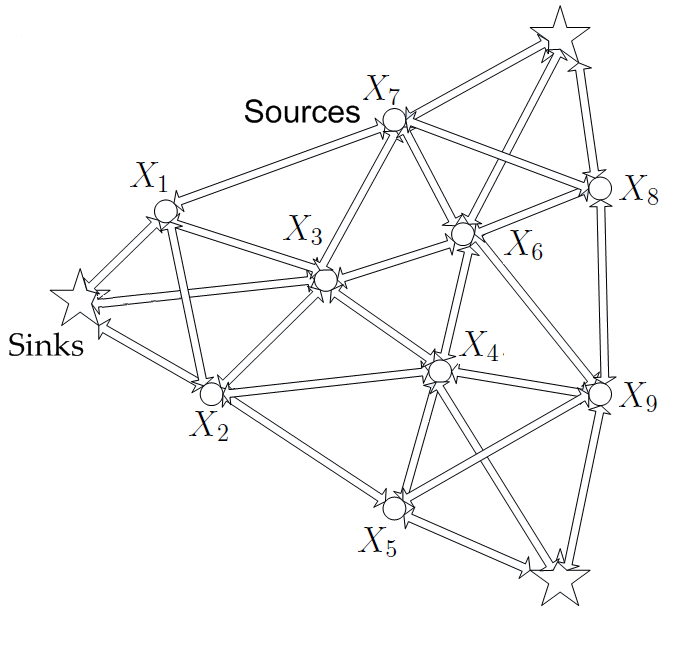}\caption{A general multi-source multi-sink sensor network. The circles denote sources and stars denote sinks. The arrows denote allowed communication links.\label{fig:A-general-multi-sink}}

\end{figure}

This paper focuses on minimum cost communication of correlated sources
over a network with multiple-sinks. We introduce a new routing paradigm
called Dispersive Information Routing (DIR), wherein intermediate
nodes are allowed to ``split a packet'' and forward a subset of
the received bits on each of the forward paths. This paradigm opens
up a rich class of research problems which focus on the interplay
between encoding and routing in a network. What makes it particularly
interesting is the challenge in encoding sources such that exactly
the required information is routed to each sink, to reconstruct the
 prescribed subset of sources. We will show, using simple examples
that asymptotically, DIR achieves a lower cost over conventional routing methods,
wherein the sinks usually receive more information than they need.
This paradigm leads to a general class of information theoretic problems,
which have not been studied earlier. In this paper, we formulate this
problem and the associated rate region. We introduce a new (random)
coding technique using principles from multiple descriptions encoding
and Han and Kobayashi decoding, which leads to an achievable rate
region for this problem. We show that this achievable rate region
is complete under certain special scenarios.

The rest of the paper is organized as follows. In Section \ref{sec:Prior_results},
we review prior work related to distributed source coding and network
routing. Before stating the problem formally, in Section \ref{sec:DIR_simple_net},
we provide 2 simple examples to demonstrate the basic principles behind
DIR and the new encoding scheme. We also demonstrate the suboptimality
of conventional routing methods using these simple examples. In Section
\ref{sec:General_setup}, we formally state the DIR problem and provide
an achievable rate region. Finally, in Section \ref{sec:Outerbounds-to-certain},
we show that this achievable rate region is complete for some special
cases of the setup.

\section{Prior Work\label{sec:Prior_results}}

Multi-terminal source coding has one of its early roots in the seminal
work of Slepian and Wolf \cite{Slepian_Wolf}. They showed, in the
context of lossless coding, that side-information available only at
the decoder can nevertheless be fully exploited as if it were available
to the encoder, in the sense that there is no asymptotic performance
loss. Later, Wyner and Ziv \cite{Wyner_Ziv} derived a lossy coding
extension that bounds the rate-distortion performance in the presence
of decoder side information. Extensive work followed considering different
network scenarios and obtaining achievable rate regions for them,
including \cite{berger_multi_term,Cover,EGC,Tung,Wagner,Wyner_CI,ZB,GW,Wyner_Source_coding}.
Han and Kobayashi \cite{Han_Kobayashi} extended the Slepian-Wolf
result to general multi-terminal source coding scenarios. For a multi-sink
network, with each sink reconstructing a prespecified subset of the sources, they characterized
an achievable rate region for lossless reconstruction of the required 
sources at each sink. Csisz$\mathrm{\acute{a}}$r and K$\mathrm{\ddot{o}}$rner \cite{Korner_Source_Networks}
provided an alternative characterization of the achievable rate region for the same setup by relating the region to the solution of a class of problems called the ``entropy characterization problems''.

There has also been a considerable amount of work on joint compression-routing
for networks. A survey of routing techniques for sensor networks is
given in \cite{Luo}. It was shown in \cite{NP_Completeness} that
the problem of finding the optimum route for compression using explicit
communication is an NP-complete problem. \cite{pattam_08} compared
different joint compression-routing schemes for a correlated sensor
grid and also proposed an approximate, practical, static source clustering
scheme to achieve compression efficiency. Much of the above work is
related to compression using explicit communication, without recourse
to distributed source coding techniques. Cristescu et al. \cite{Networked_Slepian_Wolf}
considered joint optimization of Slepian-Wolf coding and a routing
mechanism, we call `broadcasting'\footnote{Note that we loosely use the term `broadcasting' instead of `multicasting'
to stress the fact that \textbf{\textit{all}} the information transmitted
by any source is routed to every sink that reconstructs the source.
Also, our approach to routing is in some aspects, a variant of multicasting.%
}, wherein each source broadcasts its information to all sinks that
intend to reconstruct it. Such a routing mechanism is motivated from
the extensive literature on optimal routing for independent sources
\cite{Cormen}. \cite{distance_entropy} proved the general optimality
of that approach for networks with a single sink. We demonstrated its sub-optimality for the multi-sink scenario, recently in \cite{ISIT10}. This paper takes a step further towards finding the best joint compression-routing
mechanism for a multi-sink network. We note that a preliminary version of our results appeared in \cite{ITW10} and \cite{ISIT_2011_DIR}.

We note the existence of a volume of work on minimum cost network
coding for correlated sources, e.g. \cite{NC_Cost,Ramamoorthy}. But
the routing mechanism we introduce in this paper does not require
possibly complex network coders at intermediate nodes, and can be
realized using simple conventional routers. The approach does have
potential implications on network coding, but these are beyond the
scope of this paper.

\section{Dispersive Information Routing - Simple Networks \label{sec:DIR_simple_net}}

\subsection{Basic Notation\label{sub:Basic-Notation}}

We begin by introducing the basic notation. In what follows, $2^{\mathcal{S}}$
denotes the set of all subsets (power set) of any set $\mathcal{S}$
and $|\mathcal{S}|$ denotes the set cardinality. Note that $|2^{\mathcal{S}}|=2^{|\mathcal{S}|}$.
$\mathcal{S}^{c}$ denotes the set complement (the universal set will
be specified when there is ambiguity) and $\phi$ denotes the null
set. For two sets $\mathcal{S}_{1}$ and $\mathcal{S}_{2}$, we denote
the set difference by $\mathcal{S}_{1}-\mathcal{S}_{2}=\{s:s\in\mathcal{S}_{1},s\notin\mathcal{S}_{2}\}$.
Random variables are denoted by upper case letters (for example $X$)
and their realizations are denoted by lower case letters (for example
$x$). We also use upper case letters to denote source nodes and sinks
and the ambiguity will be clarified wherever necessary. A sequence of $n$ independent
 and identically distributed (iid) random variables and its realization are denoted by $X^{n}$ and $x^{n}$,  respectively.
The length $n$, $\epsilon$-typical set is denoted by $\mathcal{T}_{\epsilon}^{n}$.
$X\leftrightarrow Y\leftrightarrow Z$ denotes that the three random
variables $(X,Y,Z)$ form a Markov chain in that order. Notation in
\cite{Cover-book} is used to denote standard information theoretic
quantities.

\subsection{Illustrative example - No helpers case\label{sub:Motivating-example_1}}

\begin{figure}
\centering\includegraphics[scale=0.3]{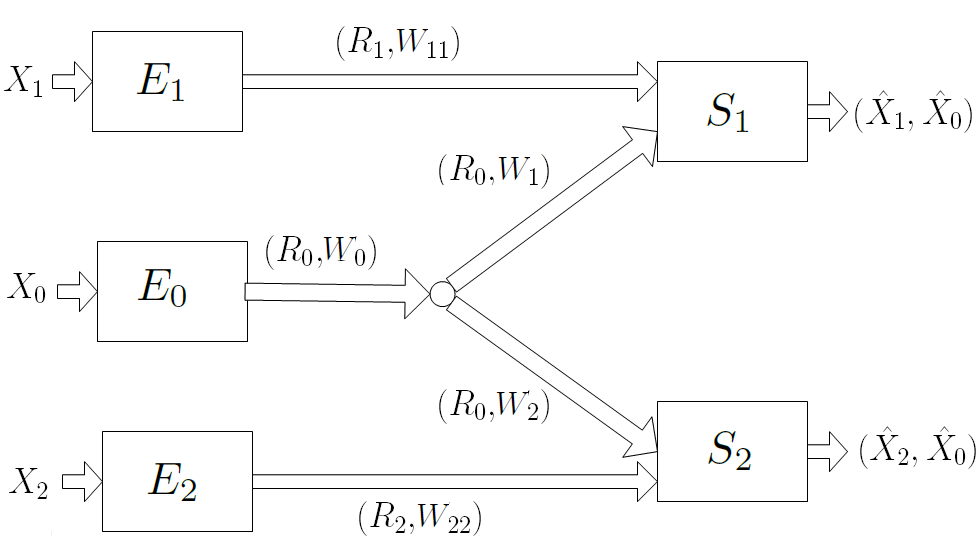}\caption{Example 1 - Conventional Routing\label{fig:eg1_convr}}

\end{figure}

Consider the network shown in Fig. \ref{fig:eg1_convr}. There are
three source nodes, $E_{0}$, $E_{1}$ and $E_{2}$ and two sinks $S_{1}$
and $S_{2}$. The three source nodes observe correlated memoryless sequences $X_{0}^n,X_{1}^n$ and $X_{2}^n$, respectively.
Sink $S_{1}$ reconstructs the pair $(X_{0}^n,X_{1}^n)$, while $S_{2}$
reconstructs $(X_{0}^n,X_{2}^n)$. $E_{0}$ communicates with the
two sinks through an intermediate node (called the `collector') which
is functionally a simple router. The edge weights on each path in
the network are as shown in the figure. The cost of communication through
an edge, $e$, is a function of the bit rate flowing through it, denoted
by $R_{e}$ and the corresponding edge weight, denoted by $W_{e}$,
which in this paper, we will assume for simplicity to be a simple product $C(R_{e},W_{e})=R_{e}W_{e}$, noting that the approach is directly extendible to
more complex cost functions. We further assume that the total cost
is the sum of individual communication cost over each edge. The objective
is to find the minimum total communication cost for lossless transmission 
of sources to the respective sinks.

We first consider the communication cost when broadcast routing is
employed \cite{Networked_Slepian_Wolf} wherein the routers forward
all the bits received from a source to all the decoders that would
reconstruct it. In other words, routers are not allowed to {}``split''
a packet and forward a portion of the received information on the
forward paths. Hence the branches connecting the collector to the
two sinks carry the same rates as the branch connecting $E_{0}$
to the collector. We denote the rate at which $X_{0}$, $X_{1}$ and
$X_{2}$ are encoded by $R_{0}$, $R_{1}$ and $R_{2}$, respectively.

Using results in \cite{Networked_Slepian_Wolf}, it can be shown that
the minimum communication cost under broadcast routing is given by
the solution to the following linear programming formulation:\begin{equation}
C_{br}=\min\{(W_{0}+W_{1}+W_{2})R_{0}+W_{11}R_{1}+W_{22}R_{2}\}\label{eq:conv_obj}\end{equation}
 under the constraints:\begin{eqnarray}
R_{0} & \geq & \max(H(X_{0}|X_{1}),H(X_{0}|X_{2}))\nonumber \\
R_{1} & \geq & H(X_{1}|X_{0})\nonumber \\
R_{2} & \geq & H(X_{2}|X_{0})\nonumber \\
R_{1}+R_{0} & \geq & H(X_{0},X_{1})\nonumber \\
R_{2}+R_{0} & \geq & H(X_{0},X_{2})\label{eq:conv_constr}\end{eqnarray}

To gain intuition into dispersive information routing, we will later
consider a special case of the above network when the branch weights
are such that $W_{11},W_{22}\ll W_{0},W_{1},W_{2}$. Let us specialize
the above equations for this case. The constraint $W_{11},W_{22}\ll W_{0},W_{1},W_{2}$,
implies that $X_{1}$ and $X_{2}$ should be encoded at rates $R_{1}=H(X_{1})$
and $R_{2}=H(X_{2})$, respectively. Therefore the scenario effectively
captures the case when $X_{1}$ and $X_{2}$ are available as side
information at the respective decoders. It follows from (\ref{eq:conv_obj})
and (\ref{eq:conv_constr}) that for achieving minimum communication
cost, $R_{0}$ is:\begin{equation}
R_{0}^{*}=\max\left\{ H(X_{0}|X_{1}),H(X_{0}|X_{2})\right\} \label{eq:no_DIR_Ro}\end{equation}
 and therefore the minimum communication cost is given by:\begin{eqnarray}
C_{br}^{*} & = & (W_{0}+W_{1}+W_{2})R_{0}^{*}\nonumber \\
 &  & +W_{11}H(X_{1})+W_{22}H(X_{2})\label{eq:no_DIR_cost}\end{eqnarray}
Is this the best we can do? The collector has to transmit enough information
to sink $S_{1}$ for it to decode $X_{0}$ and therefore the rate
is at least $H(X_{0}|X_{1})$. Similarly the rate on the branch connecting
the collector to $S_{2}$ is at least $H(X_{0}|X_{2})$. But if $H(X_{0}|X_{1})\neq H(X_{0}|X_{2})$,
there is excess rate on one of the branches.

Let us now relax this restriction and allow the collector node to
{}``split\textquotedblright{} the packet and route different subsets
of the received bits on the forward paths. We could equivalently think
of the source $E_{0}$ transmitting 3 smaller packets to the collector;
the first packet has a rate $R_{0,\{1,2\}}$ bits and is destined to both
sinks. Two other packets have rates $R_{0,1}$ and $R_{0,2}$ and
are destined to sinks $S_{1}$ and $S_{2}$, respectively. Technically,
in this case, the collector is again a simple conventional router.

We refer to such a routing mechanism, where each intermediate node transmits
a subset of the received bits on each of the forward paths, as {}``\textbf{\textit{Dispersive
Information Routing}}\textquotedblright{} (DIR). Note that unlike
network coding, DIR does not require possibly expensive coders at
intermediate nodes, and can always be realized using conventional
routers, with each source transmitting multiple packets into the network
intended to different subsets of sinks. Hereafter, we interchangeably
use the ideas of {}``packet splitting'' at intermediate nodes and
conventional routing of smaller packets, noting the equivalence in
achievable rates and costs. This scenario is depicted in Fig. \ref{fig:Example---DIR}
with the modified cost each packet encounters.

\begin{figure}
\centering\includegraphics[scale=0.3]{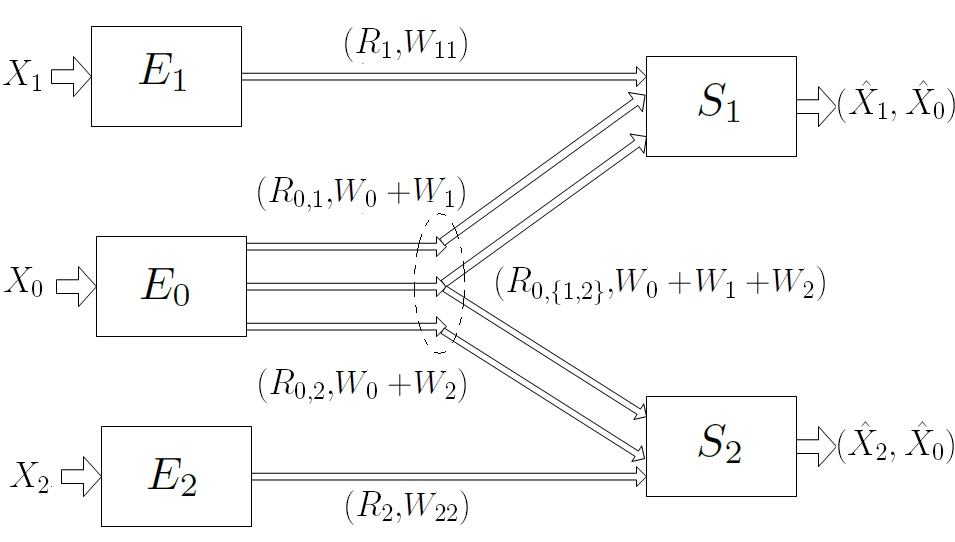}\caption{Example - DIR\label{fig:Example---DIR}. Note that the notion of `packet
splitting' is equivalent to the sources transmitting multiple smaller
packets}

\end{figure}

Two obvious questions arise - Does DIR achieve a lower communication
cost compared to conventional routing? If so, what is the minimum
communication cost under DIR?

We first aim to find the minimum cost using DIR under the special
case of $W_{11},W_{22}\ll W_{0},W_{1},W_{2}$ (i.e., $R_{1}=H(X_{1})$
and $R_{2}=H(X_{2})$). To establish the minimum communication cost
we need to first establish the complete achievable rate region for
the rate tuple $\{R_{0,1},R_{0,\{1,2\}},R_{0,2}\}$ for lossless reconstruction
of $X_{0}^n$ at both the decoders and then find the point in the achievable
rate region that minimizes the total communication cost, determined
using the modified weights shown in Fig. \ref{fig:Example---DIR}.
Before deriving the ultimate solution, it is instructive to consider one
operating point, $P_{1}\triangleq\{R_{0,1},R_{0,\{1,2\}},R_{0,2}\}=\{I(X_{1};X_{0}|X_{2}),H(X_{0}|X_{1},X_{2}),I(X_{2};X_{0}|X_{1})\}$ 
and provide the coding scheme that achieves it. Extension to other
{}``interesting points'' and to the whole achievable region follows
in similar lines. This particular rate point is considered first due
to its intuitive appeal as shown in a Venn diagram (Fig. \ref{fig:Venn-Diagram-based}a).

\begin{figure}
~~~~\includegraphics[scale=0.2]{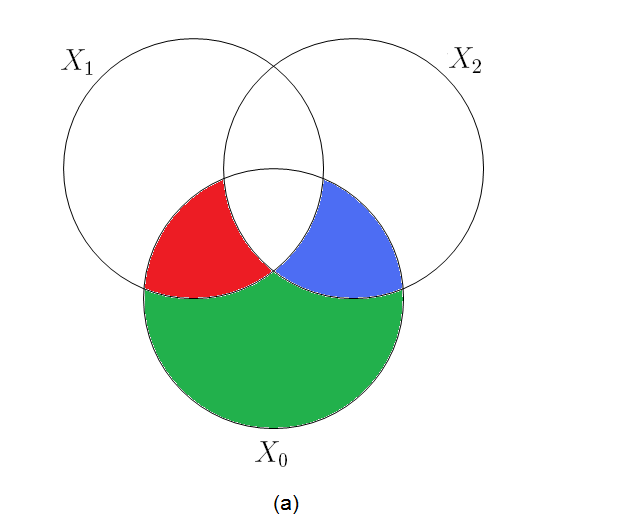}~~~~~~\includegraphics[scale=0.2]{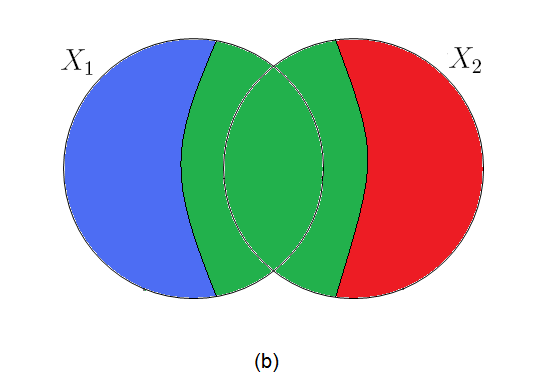}\caption{Venn Diagram based intuition: (a) Amount of information routed using DIR when operating at point $P_1$. Observe that each of the sinks receive information at the respective minimum rates. Green represents $R_{0,12}$, Blue represents $R_{0,1}$ and Red represents $R_{0,2}$ (b) Intuitive representation of Wyner's common information. Observe that in Wyner's setup, it is generally not possible to split the information exactly and that there is a rate loss due to transmitting the common bit stream.\label{fig:Venn-Diagram-based}}

\end{figure}

Gray and Wyner considered a closely resembling network \cite{GW}
shown in Fig. \ref{fig:Wyner's-Setup}. In their setup, the encoder
observes iid sequences of $2$ correlated random variables $(X_{1},X_{2})$ and transmits $3$
packets (at rates $R_{0,1},R_{0,\{1,2\}},R_{0,2}$, respectively),
one meant for each subset of sinks. The two sinks reconstruct sequences 
$X_{1}^n$ and $X_{2}^n$, respectively. They showed that the rate tuple
$\{R_{0,1},R_{0,\{1,2\}},R_{0,2}\}=\{H(X_{1}|X_{2}),I(X_{1};X_{2}),H(X_{2}|X_{1})\}$
is not achievable in general and that there is a rate loss due to
transmitting a common bit stream; in the sense that individual decoders
must receive more information than they need to reconstruct their
respective sources if the sum rate is maintained at minimum. Wyner
defined the term {}``Common Information'' \cite{Wyner_CI}, here
denoted by $C_{W}(X_{1};X_{2})$ as the minimum rate $R_{0,\{1,2\}}$
such that $\{R_{0,1},R_{0,\{1,2\}},R_{0,2}\}$ is achievable and $R_{0,1}+R_{0,\{1,2\}}+R_{0,2}=H(X_{1},X_{2})$.
He also showed that $C(X_{1};X_{2})=\min I(X_{1},X_{2};U)$ where
the $\min$ is taken over all auxiliary random variables $U$ such
that $X_{1}\leftrightarrow U\leftrightarrow X_{2}$ form a Markov
chain. He further showed that, in general, $I(X_{1};X_{2})\leq C_{w}(X_{1};X_{2})\leq\max(H(X_{1}),H(X_{2}))$.
We note in passing, the existence of an earlier definition of common information by G$\mathrm{\acute{a}}$cs and  K$\mathrm{\ddot{o}}$rner \cite{Korner_CI}
which measures the maximum shared information that can be fully utilized
by both the decoders. It is less relevant to dispersive information
routing.

\begin{figure}
\centering\includegraphics[scale=0.25]{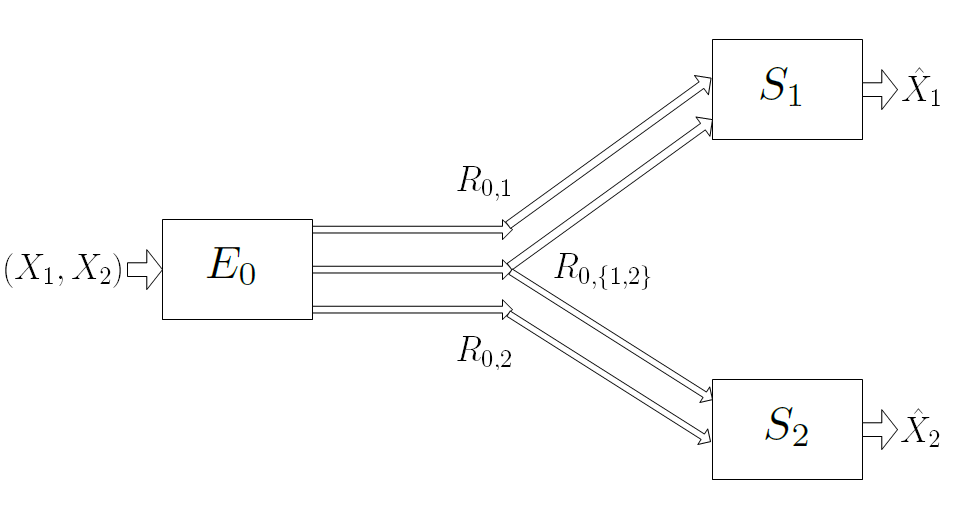}\caption{Gray-Wyner Setup\label{fig:Wyner's-Setup}. Note the resemblance to the
DIR setup in Fig. \ref{fig:Example---DIR}}

\end{figure}

At first glance, it might be tempting to extend Wyner's argument to
the DIR setting and say $P_{1}$ is not achievable in general, i.e.,
each decoder has to receive more information than it needs. But interestingly
enough, a rather simple coding scheme achieves this point and simple
extensions of the coding scheme can achieve the entire rate region
for this example. The primary difference between Gray-Wyner network
and DIR is that in their setup two correlated sources are encoded
jointly for separate decoding at each sink. However, in our setup,
 $X_{0}^n$ is encoded for lossless decoding at both the sinks.
Note that this section only provides intuitive arguments to
support the result. A coding scheme will be formally derived in section
\ref{sec:General_setup}, for the general setup.

We concentrate on encoding at $E_{0}$ assuming that $E_{1}$ and $E_{2}$ transmit at their respective source entropies.
$E_{0}$ observes a sequence of $n$ iid random variables $X_{0}^n$. This sequence belongs to the typical set, $\mathcal{T}_{\epsilon}^{n}$,
with high probability. Every typical sequence is assigned $3$ indices,
each independent of the other. The three indices are assigned using
uniform pmfs over $[1:2^{nR_{_{0,1}}}]$, $[1:2^{nR_{0,\{1,2\}}}]$
and $[1:2^{nR_{0,2}}]$, respectively. All the sequences with the same
first index, $m_{0,1}$, form a bin $\mathcal{B}_{0,1}(m_{0,1})$. Similarly
bins $\mathcal{B}_{0,2}(m_{0,2})$ and $\mathcal{B}_{0,12}(m_{0,12})$ are formed
for all indices $m_{0,2}$ and $m_{0,12}$, respectively. Upon observing a sequence $x_{0}^{n}\in\mathcal{T}_{\epsilon}^{n}$
with indices $m_{0,1},m_{0,2}$ and $m_{0,12}$, the encoder transmits index
$m_{0,1}$ to decoder $1$ alone, index $m_{0,2}$ to decoder $2$ alone
and index $m_{0,12}$ to both the decoders.

The first decoder receives indices $m_{0,1}$ and $m_{0,12}$. It tries
to find a typical sequence $\hat{x}_{0}^{n}\in\mathcal{B}_{0,1}(m_{0,1})\cap\mathcal{B}_{0,12}(m_{0,12})$
which is jointly typical with the decoded information sequence $x_{1}^{n}$.
As the indices are assigned independent of each other, every typical
sequence has uniform pmf of being assigned to the index pair $\{m_{0,1},m_{0,12}\}$
over $[1:2^{n(R_{0,1}+R_{0,\{1,2\}})}]$. Therefore, having received
indices $m_{0,1}$ and $m_{0,12}$, using arguments similar to Slepian-Wolf \cite{Slepian_Wolf} 
and Cover \cite{Cover}, the probability of decoding
error asymptotically approaches zero if:\begin{equation}
R_{0,1}+R_{0,\{1,2\}}\geq H(X_{0}|X_{1})\label{eq:eg1_ach_1}\end{equation}
 Similarly, probability of decoding error approaches zero at the second
decoder if:\begin{equation}
R_{0,2}+R_{0,\{1,2\}}\geq H(X_{0}|X_{2})\label{eq:eg1_ach_2}\end{equation}
 Clearly (\ref{eq:eg1_ach_1}) and (\ref{eq:eg1_ach_2}) imply that
$P_{1}$ is achievable. In similar lines to \cite{Slepian_Wolf,Cover},
the above achievable region can also be shown to satisfy the converse
and hence is the complete achievable rate region for this problem.
We term such a binning approach as `Power Binning' as an independent
index is assigned to each (non-trivial) subset of the decoders - the power
set. It is worthwhile to note that the same rate region can be obtained
by applying results of Han and Kobayashi \cite{Han_Kobayashi}, assuming
3 independent encoders at $E_{0}$, albeit with a more complicated
coding scheme involving multiple auxiliary random variables (see also
\cite{GW_side}). We also note that the mechanism of assigning multiple
independent random bin indices has been used is several related prior
work, such as \cite{Wyner_satellite,BC_partial_receiver_side}.

The minimum cost operating point is the point that satisfies equations
(\ref{eq:eg1_ach_1}) and (\ref{eq:eg1_ach_2}) and minimizes the
cost function:\begin{eqnarray}
C_{DIR-SI}^{*} & = & \min\,\,\,\,\{(W_{0}+W_{1})R_{0,1}+(W_{0}+W_{2})R_{0,2}\nonumber \\
 &  & +(W_{0}+W_{1}+W_{2})R_{0,\{1,2\}}\}\label{eq:eg1_min_cost}\end{eqnarray}

The solution is either one of the two points $P_{2}\triangleq\{0,H(X_{0}|X_{1}),H(X_{0}|X_{2})-H(X_{0}|X_{1})\}$
or $P_{3}\triangleq\{H(X_{0}|X_{1})-H(X_{0}|X_{2}),H(X_{0}|X_{2}),0\}$ and
both achieve lower total communication cost compared to broadcast
routing, $C_{conv}^{*}$ in (\ref{eq:no_DIR_cost}), for any
$W_{0},W_{1},W_{2}\gg W_{11},W_{22}$ if $H(X_{0}|X_{1})\neq H(X_{0}|X_{2})$.

The above coding scheme can be easily extended to the case of arbitrary
edge weights. Then, the rate region for the tuple $\{R_{1},R_{2},R_{0,1},R_{0,\{1,2\}},R_{0,2}\}$
and the cost function to be minimized are given by:

\begin{eqnarray}
C_{DIR}^{*}=\min\,\,\,\,\{W_{11}R_{1}+W_{22}R_{2}+(W_{0}+W_{1})R_{0,1}\nonumber \\
+(W_{0}+W_{2})R_{0,2}+(W_{0}+W_{1}+W_{2})R_{0,\{1,2\}}\}\label{eq:min_cost_eg2}\end{eqnarray}
 under the constraints:\begin{eqnarray}
R_{1} & \geq & H(X_{1}|X_{0})\nonumber \\
R_{0,1}+R_{0,\{1,2\}} & \geq & H(X_{0}|X_{1})\nonumber \\
R_{1}+R_{0,1}+R_{0,\{1,2\}} & \geq & H(X_{0},X_{1})\nonumber \\
R_{2} & \geq & H(X_{2}|X_{0})\nonumber \\
R_{0,2}+R_{0,\{1,2\}} & \geq & H(X_{0}|X_{2})\nonumber \\
R_{2}+R_{0,2}+R_{0,\{1,2\}} & \geq & H(X_{0},X_{2})\label{eq:min_cost_eg2_constraints}\end{eqnarray}
 If $R_{1}=H(X_{1})$ and $R_{2}=H(X_{2})$, (\ref{eq:min_cost_eg2_constraints})
specializes to (\ref{eq:eg1_ach_1}) and (\ref{eq:eg1_ach_2}). Also,
it can easily be shown that the total communication cost obtained
as a solution to the above formulation is lower than that for
conventional routing if $W_{0},W_{1},W_{2}>0$. This example clearly
demonstrates the gains of DIR over broadcast routing to communicate
correlated sources over a network.

Observe that in the above example, the sinks only receive information
from the source nodes they intend to reconstruct. Such a scenario
is called the `No helpers' case in the literature \cite{Korner_Source_Networks}.
In a network with multiple sources and sinks, if source $i$ is to
be reconstructed at a subset of sinks $\Pi_{i}$, power binning assigns
$2^{|\Pi_{i}|}-1$ independently generated indices, each being routed
to a subset of $\Pi_{i}$. It will be shown later in section \ref{sec:Outerbounds-to-certain}
that power binning achieves minimum cost under DIR, even for a general
 setup, as long as there are no helpers, i.e., when each sink is allowed to receive information
only from the requested sources. However, the problem of establishing
the complete achievable rate region becomes considerably harder when
every source is allowed to communicate with every sink, a scenario, that is highly relevant to practical networks. It was shown in \cite{ISIT10} that for certain
networks, unbounded gains in communication cost are obtained when
 source nodes are allowed to communicate with sinks that do not reconstruct them. In this paper, we derive an achievable rate
region for this setup. In the following subsection, to keep the notations
and understanding simple, we begin with one of the simplest setups
which illustrates the underlying ideas.

\subsection{A simple network with helpers\label{sec:2_s_2_s}}

\begin{figure}
\centering\includegraphics[scale=0.3]{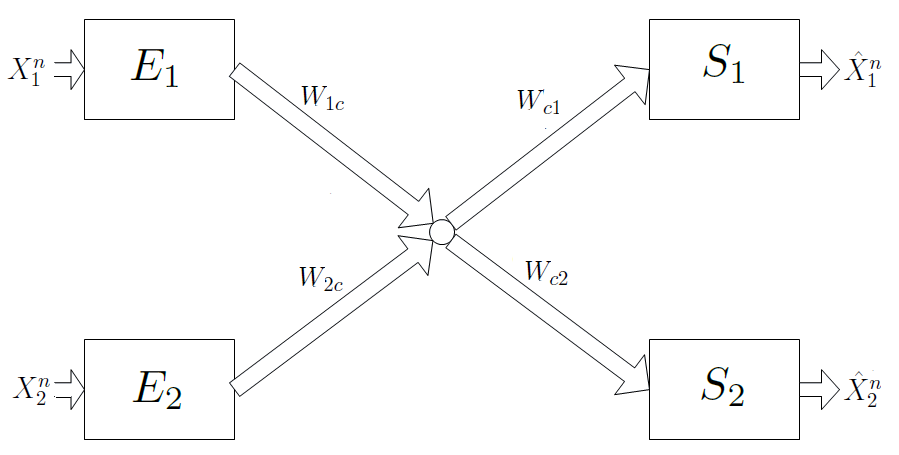}\caption{The 2 Source - 2 Sink example. Each source acts as the principle source
for one sink and as a helper for the other\label{fig:2_s_2_s}}

\end{figure}

We will again provide only intuitive description for the encoding
scheme here and defer the formal proofs for the general case to section
\ref{sec:General_setup}. Consider the network shown in Fig. \ref{fig:2_s_2_s}.
Two source nodes $E_{1}$ and $E_{2}$ observe correlated memoryless sequences $X_{1}^n$ and $X_{2}^n$, respectively. Two sinks
$S_{1}$ and $S_{2}$ require lossless reconstructions of $X_{1}^n$
and $X_{2}^n$, respectively. The source nodes can communicate with the sinks
only through a collector node. The edge weights are as shown in the
figure. Observe that, each source, while requested by one sink, acts as helper for the other.

Under dispersive information routing, each source transmits a packet
to every subset of sinks. In this example, $E_{1}$ sends 3
packets to the collector at rates $(R_{1,1},R_{1,2},R_{1,12})$, respectively.
The collector forwards the first packet to $S_{1}$, the second
to $S_{2}$ and the third to both $S_{1}$ and $S_{2}$. Similarly,
 $E_{2}$ sends 3 packets to the collector at rates $(R_{2,1},R_{2,2},R_{2,12})$
which are forwarded to the corresponding sinks. Our objective is to
determine the set of achievable rate tuples $(R_{1,1},R_{1,2},R_{1,12},R_{2,1},R_{2,2},R_{2,12})$
 that allows for lossless reconstruction at the two sinks. The minimum cost then follows by finding the point in the
achievable rate region which minimizes the effective communication
cost, $C_{DIR}$, given by:\begin{eqnarray}
\sum_{i=1}^{2}(W_{ic}+W_{c1}+W_{c2})R_{i,12}+(W_{2c}+W_{c1})R_{2,1}&  & \nonumber \\
 +\sum_{i=1}^{2}(W_{ic}+W_{ci})R_{i,i}+(W_{1c}+W_{c2})R_{1,2}\label{eq:2_E_2_E_cost}&  &\end{eqnarray}

A non-single letter characterization of the complete rate region is
possible using the results of Han and Kobayashi in \cite{Han_Kobayashi}.
They also provide a single-letter partial achievable rate region.
However, applicability of their result requires artificial imposition of 3 independent encoders at each source,
which is an unnecessary restriction. We present a more general achievable
rate region, which maintains the dependencies between the messages
at each encoder. Note that the source coding setup which arises out
of the DIR framework is a special case of the general problem of distributed
multiple descriptions and therefore the principles underlying the
coding schemes for distributed source coding \cite{Han_Kobayashi}
and multiple descriptions encoding \cite{VKG} play crucial roles
in deriving a coding mechanism for dispersive information routing. It is interesting to observe that, unlike the
general MD setting, the DIR framework is non-trivial even in the lossless
scenario and deriving a complete rate region for lossless reconstruction
at all the sinks is a challenging problem.

We now give an achievable region for the example in Fig. \ref{fig:2_s_2_s}.
Suppose we are given random variables $(U_{1,12},U_{1,1},U_{1,2},U_{2,12},U_{2,1},U_{2,2})$
jointly distributed with $(X_{1},X_{2})$ such that the following
Markov chain conditions hold:\begin{eqnarray}
(U_{1,12},U_{1,1}) & \leftrightarrow & X_{1}\leftrightarrow X_{2}\leftrightarrow(U_{2,12},U_{2,1})\nonumber \\
(U_{1,12},U_{1,2}) & \leftrightarrow & X_{1}\leftrightarrow X_{2}\leftrightarrow(U_{2,12},U_{2,2})\label{eq:2_E_2_E_Markov}\end{eqnarray}
 Note that the codeword indices of $U_{i,\mathcal{S}}$ are sent in
the packet from source $E_{i}$ to sinks $S_{j}:j\in\mathcal{S}$.
The encoding is divided into 3 stages.

\textit{Encoding} : We first focus on the encoding at $E_{1}$. In
the first stage, $2^{nR_{1,12}^{'}}$ codewords of $U_{1,12}$, each
of length $n$ are generated independently, with elements drawn according
to the marginal density $P(U_{1,12})$. Conditioned on each of these
codewords, $2^{nR_{1,1}^{'}}$ and $2^{nR_{1,2}^{'}}$ codewords of
$U_{1,1}$ and $U_{1,2}$ are generated according to the conditional
densities $P(U_{1,1}|U_{1,12})$ and $P(U_{1,2}|U_{1,12})$, respectively.
Codebooks for $U_{2,12},U_{2,1}$ and $U_{2,2}$ are generated at
$E_{2}$ in a similar fashion. On observing a sequence $x_{1}^{n}$,
$E_{1}$ first tries to find a codeword tuple from the codebooks of
$(U_{1,12},U_{1,1},U_{1,2})$ such that $(x_{1}^{n},u_{1,12}^{n},u_{1,1}^{n})\in\mathcal{T}_{\epsilon}^{n}$
and $(x_{1}^{n},u_{1,12}^{n},u_{1,2}^{n})\in\mathcal{T}_{\epsilon}^{n}$.
The probability of finding such a codeword tuple approaches 1 if,
\begin{eqnarray}
R_{1,12}^{'} & \geq & I(X_{1};U_{1,12})\nonumber \\
R_{1,1}^{'} & \geq & I(X_{1};U_{1,1}|U_{1,12})\nonumber \\
R_{1,2}^{'} & \geq & I(X_{1};U_{1,2}|U_{1,12})\label{eq:2_E_2_E_0}\end{eqnarray}

Let the codewords selected be denoted by ($u_{1,12},$$u_{1,1}$,$u_{1,2}$).
Similar constraints on $(R_{2,1}^{'},R_{2,2}^{'},R_{2,12}^{'})$ must
be satisfied for encoding at $E_{2}$. Denote the codewords selected
at $E_{2}$ by $(u_{2,12},u_{2,1},u_{2,2})$. It follows from (\ref{eq:2_E_2_E_Markov})
and the `Conditional Markov Lemma' in \cite{Wagner} that $(x_{1}^{n},x_{2}^{n},u_{1,12},u_{1,1},u_{2,12},u_{2,1})\in\mathcal{T}_{\epsilon}^{n}$
and $(x_{1}^{n},x_{2}^{n},u_{1,12},u_{1,2},u_{2,12},u_{2,2})\in\mathcal{T}_{\epsilon}^{n}$
with high probability.

In the second stage of encoding, each encoder uniformly divides the
$2^{nR_{i,\mathcal{S}}^{'}}$ codewords of $U_{i,\mathcal{S}}$ into
$2^{nR_{i,\mathcal{S}}^{''}}$ bins $\forall\, i\in\{1,2\}$, $\mathcal{S}\in\{1,2,12\}$.
All the codewords which have the same bin index $m$ are said to fall
in the bin $\mathcal{C}_{i,\mathcal{S}}(m)$ $\forall m\in(1\ldots2^{nR_{i,\mathcal{S}}^{''}})$.
Note that the number of codewords in bin $\mathcal{C}_{i,\mathcal{S}}(m)$
is $2^{n(R_{i,\mathcal{S}}^{'}-R_{i,\mathcal{S}}^{''})}$. If $E_{1}$
selects the codewords $(u_{1,12},u_{1,1},u_{1,2})$ in the first stage
and if the bin indices associated with $(u_{1,12},u_{1,1},u_{1,2})$
are $(m_{1,12},m_{1,1},m_{1,2})$, then index $m_{1,1}$ is routed
to sink $S_{1}$, $m_{1,2}$ to sink $S_{2}$ and $m_{1,12}$ to both
the sinks $S_{1}$ and $S_{2}$. Similarly, bin indices $(m_{2,12},m_{2,1},m_{2,2})$
are routed from $E_{2}$ to the corresponding sinks.

The third stage of encoding, resembles the `Power Binning' scheme
described in Section \ref{sub:Motivating-example_1}. Every typical
sequence of $X_{1}^{n}$ is assigned a random bin index uniformly
chosen over $[1:2^{n\tilde{R}_{1,1}}]$. All sequences with the same
index, $l_{1,1}$, form a bin $\mathcal{B}_{1,1}(l_{1,1})$ $\forall l_{1,1}\in\{1\ldots2^{n\tilde{R}_{1,1}}\}$.
Upon observing a sequence $X_{1}^{n}\in\mathcal{T}_{\epsilon}^{n}$
with bin index $l_{1,1}$, in addition to $m_{1,1}$ (from the second
stage of encoding), encoder $E_{1}$ also routes index $l_{1,1}$
to sink $S_{1}$. Similarly bin index $l_{2,2}$ is routed from $E_{2}$
to $S_{2}$ in addition to $m_{2,2}$. These bin indices are used
to reconstruct $X_{1}^{n}$ and $X_{2}^{n}$ losslessly at the respective decoders.
Note that, in a general setup, if source $i$ is to be reconstructed
at a subset of sinks $\Pi_{i}$, the source assigns $2^{|\Pi_{i}|}-1$
independently generated indices, each being routed to a subset of
$\Pi_{i}$. We also note that $U_{1,1}$ and $U_{2,2}$ can be conveniently
set to constants without changing the overall rate region. However,
we continue to use them to avoid complex notation.

\textit{Decoding} : We again focus on the first sink $S_{1}$. It
receives the indices $(m_{1,12},m_{1,1},m_{2,12},m_{2,1},l_{1,1})$.
It first looks for a pair of unique codewords from $\mathcal{C}_{1,12}(m_{1,12})$
and $\mathcal{C}_{2,12}(m_{2,12})$ which are jointly typical. Obviously,
there is at least one pair, $(u_{1,12},u_{2,12})$, which is jointly
typical. The probability that no other pair of codewords are jointly
typical approaches $1$ if:\begin{eqnarray}
(R_{1,12}^{'}-R_{1,12}^{''})+(R_{2,12}^{'}-R_{2,12}^{''})\leq I(U_{1,12};U_{2,12})\label{eq:2_E_2_E_1}\end{eqnarray}
 Noting that $(R_{1,12}^{'}-R_{1,12}^{''})\geq0$ and $(R_{2,12}^{'}-R_{2,12}^{''})\geq0$,
and applying the constraints on $R_{1,12}^{'}$ and $R_{2,12}^{'}$
from (\ref{eq:2_E_2_E_0}) we get the following constraints for $R_{1,12}^{''}$
and $R_{2,12}^{''}$:\begin{eqnarray}
R_{1,12}^{''} & \geq & I(X_{1};U_{1,12}|U_{2,12})\nonumber \\
R_{2,12}^{''} & \geq & I(X_{2};U_{2,12}|U_{1,12})\nonumber \\
R_{1,12}^{''}+R_{2,12}^{''} & \geq & I(X_{1},X_{2};U_{1,12},U_{2,12})\label{eq:2_E_2_E_2}\end{eqnarray}
 The decoder at $S_{1}$ next looks at the codebooks of $U_{1,1}$
and $U_{2,1}$ which were generated conditioned on $u_{1,12}$ and
$u_{2,12}$, respectively, to find a unique pair of codewords from $\mathcal{C}_{1,1}(m_{1,1})$
and $\mathcal{C}_{2,1}(m_{2,1})$ which are jointly typical with $(u_{1,12},u_{2,12})$.
We again have one pair, $(u_{1,1},u_{2,1})$, which is jointly typical
with $(u_{1,12},u_{2,12})$. It can be shown using arguments similar
to \cite{Han_Kobayashi} that the probability of finding no other
jointly typical pair approaches $1$ if :\begin{eqnarray}
(R_{1,1}^{'}-R_{1,1}^{''}) & \leq & I(U_{1,1};U_{2,1},U_{2,12}|U_{1,12})\nonumber \\
(R_{2,1}^{'}-R_{2,1}^{''}) & \leq & I(U_{2,1};U_{1,1},U_{1,12}|U_{2,12})\nonumber \\
\bigl\{(R_{1,1}^{'}-R_{1,1}^{''}) & \leq & H(U_{1,1}|U_{1,12})+H(U_{2,1}|U_{2,12})\nonumber \\
+(R_{2,1}^{'}-R_{2,1}^{''})\bigr\} &  & -H(U_{1,1},U_{2,1}|U_{1,12},U_{2,12})\label{eq:2_E_2_E_31}\end{eqnarray}
 On substituting the constraints for $R_{1,1}^{'}$ and $R_{1,2}^{'}$
from (\ref{eq:2_E_2_E_0}), and using the Markov chain condition in
(\ref{eq:2_E_2_E_Markov}) we get:

\begin{eqnarray}
R_{1,1}^{''} & \geq & I(X_{1};U_{1,1}|U_{1,12},U_{2,12},U_{2,1})\nonumber \\
R_{2,1}^{''} & \geq & I(X_{2};U_{2,1}|U_{1,12},U_{2,12},U_{1,1})\nonumber \\
R_{1,1}^{''}+R_{2,1}^{''} & \geq & I(X_{1},X_{2};U_{1,1},U_{2,1}|U_{1,12},U_{2,12})\label{eq:2_E_2_E_3}\end{eqnarray}
 After successfully decoding the codewords $(u_{1,12},u_{1,1},u_{2,12},u_{2,1})$,
the decoder at $S_{1}$ looks for a unique sequence from $\mathcal{B}_{1,1}(l_{1,1})$
which is jointly typical with $(u_{1,12},u_{1,1},u_{2,12},u_{2,1})$.
We again have $x_{1}^{n}$ satisfying this property. It can be shown
that the probability of finding no other sequence which is jointly
typical with $(u_{1,12},u_{1,1},u_{2,12},u_{2,1})$ approaches $1$
if:\begin{equation}
\tilde{R}_{1,1}\geq H(X_{1}|U_{1,12},U_{2,12},U_{1,1},U_{2,1})\label{eq:2_E_2_E_4}\end{equation}
 Similar conditions at sink $S_{2}$ lead to the following constraints:

\begin{eqnarray}
R_{2,2}^{''} & \geq & I(X_{2};U_{2,2}|U_{1,12},U_{2,12},U_{1,2})\nonumber \\
R_{1,2}^{''} & \geq & I(X_{2};U_{1,2}|U_{1,12},U_{2,12},U_{2,2})\nonumber \\
R_{2,2}^{''}+R_{1,2}^{''} & \geq & I(X_{1},X_{2};U_{2,2},U_{1,2}|U_{1,12},U_{2,12})\nonumber \\
\tilde{R}_{2,2} & \geq & H(X_{2}|U_{1,12},U_{2,12},U_{1,2},U_{2,2})\label{eq:2_E_2_E_5}\end{eqnarray}

The first packet from $E_{1}$, destined to only $S_{1}$, carries
indices $(m_{1,1},l_{1,1})$ at rate $R_{1,1}=R_{1,1}^{''}+\tilde{R}_{1,1}$.
The second and third packets carry $m_{1,2}$ and $m_{1,12}$ at rates
$R_{1,2}=R_{1,2}^{''}$ and $R_{1,12}=R_{1,12}^{''}$, respectively 
and are routed to the corresponding sinks. Similarly, 3 packets are
transmitted from $E_{2}$ carrying indices $\{m_{2,1},m_{2,12},(m_{2,2},l_{2,2})\}$
at rates $(R_{2,1},R_{2,12},R_{2,2})=(R_{2,1}^{''},R_{2,12}^{''},R_{2,2}^{''}+\tilde{R}_{2,2})$
to sinks $\{S_{1},S_{2},(S_{1},S_{2})\}$, respectively. Constraints
for $(R_{1,1},R_{1,2},R_{1,12},R_{2,1},R_{2,2},R_{2,12})$ can now
be obtained using (\ref{eq:2_E_2_E_2}),(\ref{eq:2_E_2_E_3}), (\ref{eq:2_E_2_E_4})
and (\ref{eq:2_E_2_E_5}). The convex closure of achievable rates
over all such random variables $(U_{1,12},U_{1,1},U_{1,2},U_{2,12},U_{2,1},U_{2,2})$
gives the achievable rate region for the 2 source - 2 sink DIR problem.
It is easy to verify that this region subsumes the region that would be produced by employing the approach of Han and
Kobayashi \cite{Han_Kobayashi}, which must assume three independent encoders at each source. Observe that in the above illustration, we assumed that the decoding is performed in a sequential manner, i.e., the codewords of $U_{1,12}$ are decoded first followed by the codewords of $(U_{1,1})$ and $(U_{1,2})$, respectively. This was done only for the ease of understanding. In Theorem \ref{thm:main}, we derive the conditions on rates for the decoders to find typical sequences from all the codebooks jointly (at once). Note that conditions on the rates for joint decoding is generally weaker (the region is larger) than that for sequential decoding.

\section{Dispersive Information Routing - General Setup\label{sec:General_setup}}

Let a network be represented by an undirected connected graph $G=(V,\mathcal{E})$.
Each edge $e\in \mathcal{E}$ is associated with an edge weight, $W_{e}$. The
communication cost is assumed to be a simple product of the edge rate
and edge weight\footnote{The approach is applicable to more general cost functions.%
}, i.e., $C_{e}=R_{e}W_{e}$. The nodes $V$ consist of $N$ source nodes (denoted by $E_{1},E_{2}\ldots E_{N}$),
$M$ sinks (denoted by $S_{1},S_{2}\ldots S_{M}$), and $|V|\lyxmathsym{\textminus}N\lyxmathsym{\textminus}M$
intermediate nodes. We define the sets $\Sigma=\{1\ldots N\}$ and $\Pi=\{1\ldots M\}$.
Source node $E_{i}$ observes $n$ iid random variables $X_{i}^n$, each taking values 
over a finite alphabet $\mathcal{X}_{i}$. Sink $S_{j}$ reconstructs 
(requests) a subset of the sources specified by $\Sigma_{j}\subseteq\Sigma$.
Conversely, source node $E_{i}$ is reconstructed at a subset of sinks
specified by $\Pi_{i}\subseteq\Pi$. The objective is to find the minimum
communication cost achievable by dispersive information routing for
lossless reconstruction of the requested sources at each sink when
every source node can (possibly) communicate with every sink.

\subsection{Obtaining the effective costs\label{sub:Steiner}}

Under DIR each source transmits at most $2^{M}-1$ packets into the network,
each meant for a different subset of sinks. Note that, while $\Pi_i$ is the subset of sinks reconstructing $X_{i}^n$, $E_i$ may be transmitting packets to many other subsets of sinks. Let the packet from source
$E_{i}$ to the subset of sinks $\mathcal{K}\subseteq\Pi$ be denoted
by $\mathcal{P}_{i,\mathcal{K}}$ and let it carry information at
rate $R_{i,\mathcal{K}}$.

The optimum route for packet $\mathcal{P}_{i,\mathcal{K}}$ from the
source to these sinks is determined by a spanning tree optimization
(minimum Steiner tree) \cite{Cormen}. More specifically, for each
packet $\mathcal{P}_{i,\mathcal{K}}$, the optimum route is obtained
by minimizing the cost over all trees rooted at node $i$ which span
all sinks $j\in\mathcal{K}$. The minimum cost of transmitting packet
$\mathcal{P}_{i,\mathcal{K}}$ with $R_{i,\mathcal{K}}$ bits from
source $i$ to the subset of sinks $\mathcal{K}$, denoted by $d_{i}(\mathcal{K})$
is :\begin{equation}
d_{i}(\mathcal{K})=R_{i,\mathcal{K}}\min_{Q\in \mathcal{E}_{i,\mathcal{K}}}\sum_{e\in Q}w_{e}\label{eq:modified_costs}\end{equation}
 where $\mathcal{E}_{i,\mathcal{K}}$ denotes the set of all paths from source
$i$ to the subset of sinks $\mathcal{K}$. Having obtained the effective
cost for each packet in the network, our next objective is to find
an achievable rate region for the tuple $(R_{i,\mathcal{K}}\,\,\forall i\in\Sigma,\mathcal{K}\subseteq\Pi)$.
The minimum communication cost then follows directly from a simple
linear programming formulation. Note that the minimum Steiner tree
problem is NP - hard and requires approximate algorithms
to solve in practice. Also note that in theory, each encoder transmits $2^M-1$ packets into the network. While in practice we might be able to realize improvements over broadcast routing using significantly fewer packets (see e.g., \cite{EUSIPCO}).

\subsection{\label{sub:general_ach_region}An achievable rate region}

In what follows, we use the shorthand $\{U_{i}\}{}_{\mathcal{S}}$
for $\{U_{i,\mathcal{K}}:\mathcal{K}\in\mathcal{S}\}$ and $\{U_{\Gamma}\}{}_{\mathcal{S}}$
for $\{U_{i,\mathcal{K}}:i\in\Gamma,\,\mathcal{K}\in\mathcal{S}\}$.
Note the difference between $\{U_{i}\}{}_{\mathcal{S}}$ and $U_{i,\mathcal{S}}$.
$\{U_{i}\}{}_{\mathcal{S}}$ is a set of variables, whereas $U_{i,\mathcal{S}}$
is a single variable. For example, $\{U_{1}\}_{(1,2,12)}$ denotes
the set of variables $(U_{1,1},U_{1,2},U_{1,12})$ and $\{U_{(1,2)}\}_{(1,2,12)}$
represents the set $(U_{1,1},U_{1,2},U_{1,12},U_{2,1},U_{2,2},U_{2,12})$.

We first give a formal definition of a block code and an associated
rate region for DIR. We denote the set $\{1,2\ldots L\}$ by $I_{L}$
for any positive integer $L$. We assume that the source node $E_{i}$
observes the random sequence $X_{i}^n$.
An $(n,P_{e},L_{i,\mathcal{K}};\forall i\in\Sigma,\mathcal{K}\in2^{\Pi}-\phi)$
DIR-code is defined by the following mappings:
\begin{itemize}
\item \textit{Encoder}s: \begin{equation}
f_{i}^{E}:\mathcal{X}_{i}^{n}\rightarrow\prod_{\mathcal{K}\in2^{\Pi}-\phi}I_{L_{i,\mathcal{K}}}\label{eq:encoder}\end{equation}

\item \textit{Decoder}s:\begin{equation}
f_{j}^{D}:\prod_{i\in\Sigma}\prod_{\mathcal{K}\in2^{\Pi}:j\in\mathcal{K}}I_{L_{i,\mathcal{K}}}\rightarrow\{\mathcal{X}^{n}\}_{\Sigma_{j}}\label{eq:decoders}\end{equation}

\end{itemize}
Denoting $f_{i}^{E}(X_{i}^{n})=\{T_{i}\}_{2^{\Pi}-\phi}$ where $1\leq T_{i,\mathcal{K}}\leq L_{i,\mathcal{K}}$,
the decoder estimates are given by:\begin{equation}
\{\hat{X}^{n}\}_{\Sigma_{j}}=f_{j}^{D}(\{T_{\Sigma}\}_{(\mathcal{K}\in2^{\Pi}:j\in\mathcal{K})})\label{eq:estimates}\end{equation}
Note the correspondence between the encoder-decoder mappings and dispersive
information routing. Observe that packet $\mathcal{P}_{i,\mathcal{K}}$ carries $T_{i,\mathcal{K}}$ at rate $L_{i,\mathcal{K}}$  from source $i$ to the subset of sinks $\mathcal{K}$. The probability of error is defined as:\begin{equation}
P_{e}=\frac{1}{M}\left[\sum_{j\in\Pi}P(\{X^{n}\}_{\Sigma_{j}}\neq\{\hat{X}^{n}\}_{\Sigma_{j}})\right]\label{eq:Pe}\end{equation}
A rate tuple $\{R_{i,\mathcal{K}};\forall i,\mathcal{K}\}$ is said
to be achievable if for any $\eta>0$ and $0<\epsilon<1$, there exists
a $(n,P_{e},L_{i,\mathcal{K}};\forall i\in\Sigma,\mathcal{K}\in2^{\Pi}-\phi)$
code for $n$ sufficiently large such that,\begin{equation}
R_{i,\mathcal{K}}\leq\frac{1}{n}\log L_{i,\mathcal{K}}+\eta\end{equation}
with the probability of error less than $\epsilon$, i.e., \begin{equation}
P_{e}<\epsilon\end{equation}

We extend the coding scheme described in section \ref{sec:2_s_2_s}
to derive an achievable rate region for the tuple $(R_{i,\mathcal{K}}\,\,\forall i\in\Sigma,\mathcal{K}\in2^{\Pi}-\phi)$
using principles from multiple descriptions encoding \cite{VKG,EGC,ZB}
and Han and Kobayashi decoding \cite{Han_Kobayashi}, albeit with
more complex notation. Without loss of generality, we assume that
every source can send packets to every sink.

Before stating the achievable rate region in Theorem \ref{thm:main},
we define the following subsets of $2^{\Pi}$:\begin{eqnarray}
\mathcal{I}_{W} & = & \{\mathcal{K}:\mathcal{K}\in2^{\Pi},\,\,|\mathcal{K}|=W\}\nonumber \\
\mathcal{I}_{W+} & = & \{\mathcal{K}:\mathcal{K}\in2^{\Pi},\,\,|\mathcal{K}|>W\}\end{eqnarray}
Let $\mathcal{B}$ be any subset of $\Pi$ with $|\mathcal{B}|\leq W$.
We define the following subsets of $\mathcal{I}_{W}$ and $\mathcal{I}_{W+}$:\begin{eqnarray}
\mathcal{I}_{W}(\mathcal{B}) & = & \{\mathcal{K}:\mathcal{K}\in\mathcal{I}_{W},\,\,\mathcal{B}\subseteq\mathcal{K}\}\nonumber \\
\mathcal{I}_{W+}(\mathcal{B}) & = & \{\mathcal{K}:\mathcal{K}\in\mathcal{I}_{W+},\,\,\mathcal{B}\subseteq\mathcal{K}\}\end{eqnarray}
We also define: \begin{equation}
\mathcal{J}(\mathcal{S})=\{\mathcal{K}:\,\,\mathcal{K}\in2^{\Pi},\,\,|\mathcal{K}\bigcap\mathcal{S}|>0\}\end{equation}
Note that $\mathcal{J}(\Pi)=2^{\Pi}-\phi$. Let $\mathcal{Q}$ be
any subset of $2^{\Pi}-\phi$. We say that $\mathcal{Q}\in\mathcal{Q}^{*}$
if it satisfies the following property $\forall \mathcal{K} \in \mathcal{Q}$:\begin{equation}
\mbox{if }\mathcal{K}\in\mathcal{Q}\,\,\Rightarrow\,\,\,\mathcal{I}_{|\mathcal{K}|+}(\mathcal{K})\subset\mathcal{Q}\label{eq:cond_Q_main}\end{equation}Let $\{U_{\Sigma}\}{}_{\mathcal{J}(\Pi)}$ be any set of $N(2^{M}-1)$
random variables defined on arbitrary finite alphabets, jointly distributed
with $\{X\}_{\Sigma}$ satisfying the following: $\forall j\in\Pi$,\begin{equation}
P(\{X\}_{\Sigma},\{U_{\Sigma}\}_{\mathcal{J}(j)})=P(\{X\}_{\Sigma})\prod_{i\in\Sigma}P(\{U_{i}\}_{\mathcal{J}(j)}|X_{i})\label{eq:thm_Markov}\end{equation}
The above Markov condition ensures that all the codewords which reach
a sink are jointly typical with $\{X\}_{\Sigma_{j}}$.

We define $\alpha(i,\mathcal{Q})$ as:

\begin{eqnarray}
\alpha(i,\mathcal{Q}) & = & -H\left(\{U_{i}\}_{\mathcal{Q}}|X_{i}\right)\nonumber \\
 &  & +\sum_{\mathcal{K}\in\mathcal{Q}}H\left(U_{i,\mathcal{K}}|\{U_{i}\}_{\mathcal{I}_{|\mathcal{K}|+}(\mathcal{K})}\right)\label{eq:alpha_defn-1}\end{eqnarray}
 $\forall i\in\Sigma,\mathcal{Q}\subseteq\mathcal{J}(\Pi)$. We further
define $\beta(k,\mathcal{Q}_{1},\mathcal{Q}_{2},\ldots\mathcal{Q}_{N})$
$\forall k\in\Pi,\,\,\mathcal{Q}_{1},\mathcal{Q}_{2},\ldots\mathcal{Q}_{N}\subseteq\mathcal{J}(k)$
as:\begin{eqnarray}
\beta(k,\mathcal{Q}_{1},\mathcal{Q}_{2},\ldots\mathcal{Q}_{N})=H\left(\{U_{i}\}_{\mathcal{Q}_{i}^{c}}\forall i|\{U_{i}\}_{\mathcal{Q}_{i}}\forall i\right)\nonumber \\
-\sum_{i\in\Sigma}\sum_{\mathcal{K}\in\mathcal{Q}_{i}^{c}}H(U_{i,\mathcal{K}}|\{U_{i}\}_{\mathcal{I}_{|\mathcal{K}|+}(\mathcal{K})})\label{eq:bound_rate_diff-1}\end{eqnarray}
where $\mathcal{Q}_{i}^{c}=\mathcal{J}(k)-\mathcal{Q}_{i}$ and 
define $\gamma_{k}(\Gamma)$ as :\begin{eqnarray}
\gamma_{k}(\Gamma)=H\left(\{X\}_{\Gamma}|\{X\}_{\Gamma^{c}},\{U_{\Sigma}\}_{\mathcal{J}(k)}\right)\nonumber \\
\forall k\in\Pi,\Gamma\subseteq\Sigma_{k}\label{eq:gamma_defn}\end{eqnarray}
 where $\Gamma^{c}=\Sigma_{k}-\Gamma$. We state our main result in
the following Theorem.
\begin{thm}
Achievable Rate Region for DIR :\label{thm:main}Let $\{U_{\Sigma}\}{}_{2^{\Pi}-\phi}$
be any set of random variables satisfying (\ref{eq:thm_Markov}).
Let $(R_{i,\mathcal{K}}^{'}\,\,\forall i\in\Sigma,\mathcal{K}\in2^{\Pi}-\phi)$
be any set of auxiliary rate tuples such that:\begin{equation}
\sum_{\mathcal{K}\in\mathcal{Q}}R_{i,\mathcal{K}}^{'}\geq\alpha(i,\mathcal{Q})\label{eq:thm_0}\end{equation}
$\forall\mathcal{Q}\in\mathcal{Q}^{*}$. Further, let $(R_{i,\mathcal{K}}^{''}\,\,\forall i\in\Sigma,\mathcal{K}\in2^{\Pi}-\phi)$
be any set of rate tuples \textup{such that}: \begin{eqnarray}
\sum_{i\in\Sigma}\sum_{\mathcal{K}\in\mathcal{Q}_{i}^{c}}R_{i,\mathcal{K}}^{''}\geq\sum_{i\in\Sigma}\sum_{\mathcal{K}\in\mathcal{Q}_{i}^{c}}R_{i,\mathcal{K}}^{'}+\beta(k,\mathcal{Q}_{1},\mathcal{Q}_{2},\ldots\mathcal{Q}_{N})\label{eq:thm_1}\end{eqnarray}
for each $k\in\Pi$, $\forall\mathcal{Q}_{1},\mathcal{Q}_{2},\ldots\mathcal{Q}_{N}\subseteq\mathcal{J}(k)$
\textup{satisfying }(\ref{eq:cond_Q_main}) such that $\exists i\in\{1,\dots,N\}:\mathcal{Q}_{i}\neq\mathcal{J}(k)$. Let $(\tilde{R}_{i,\mathcal{K}}\,\,\forall i\in\Sigma,\mathcal{K}\in2^{\Pi_{k}}-\phi)$
satisfy: \textup{\begin{equation}
\sum_{i\in\Gamma}\sum_{\mathcal{K}:k\in\mathcal{K}}\tilde{R}_{i,\mathcal{K}}\geq\gamma_{k}(\Gamma)\label{eq:thm_2}\end{equation}
}$\forall k\in\Pi,\Gamma\in2^{\Sigma_{k}}-\phi$. Then, the achievable
rate region for the tuple \textup{$(R_{i,\mathcal{S}}\,\,\forall i\in\Sigma,\mathcal{S}\in2^{\Pi}-\phi)$
contains all rates such that,} \begin{equation}
R_{i,\mathcal{K}}\geq\begin{cases}
R_{i,\mathcal{K}}^{''}+\tilde{R}_{i,\mathcal{K}} & \mbox{if }\mathcal{K}\subseteq2^{\Pi_{i}}-\phi\\
R_{i,\mathcal{K}}^{''} & \mbox{if }\mathcal{K}\nsubseteq2^{\Pi_{i}}-\phi\end{cases}\label{eq:thm_3}\end{equation}
The convex closure of the achievable tuples over all such $N(2^{M}-1)$
random variables satisfying (\ref{eq:thm_Markov}) is the achievable
rate region for DIR and is denoted by $\mathcal{R}_{DIR}$.\end{thm}
\begin{rem}
The converse to this achievability region does not hold in general.
A simple counter example follows from the famous binary modulo two
sum problem proposed by K$\mathrm{\ddot{o}}$rner and Marton for the 2 helper setup in
\cite{mod_two_prob}. However, in section \ref{sec:Outerbounds-to-certain}
we prove the converse for certain special cases.
\end{rem}

\begin{rem}
The coding scheme in Theorem \ref{thm:main} can be easily specialized
to `power binning' by setting $\{U_{\Sigma}\}{}_{2^{\Pi}-\phi}$ to
constants. This effectively becomes the `no-helpers' scenario as setting
$\{U_{\Sigma}\}{}_{2^{\Pi}-\phi}$ to constants implies that $R_{i,\mathcal{S}}=0$
$\forall\mathcal{S}\notin2^{\Pi_{i}}$. \end{rem}
\begin{proof}
We follow the notation and the notion of strong typicality defined
in \cite{Han_Kobayashi}. We refer to \cite{Han_Kobayashi} (section
3) for formal definitions and basic Lemmas associated with typicality.

\begin{figure}
\centering\includegraphics[scale=0.4]{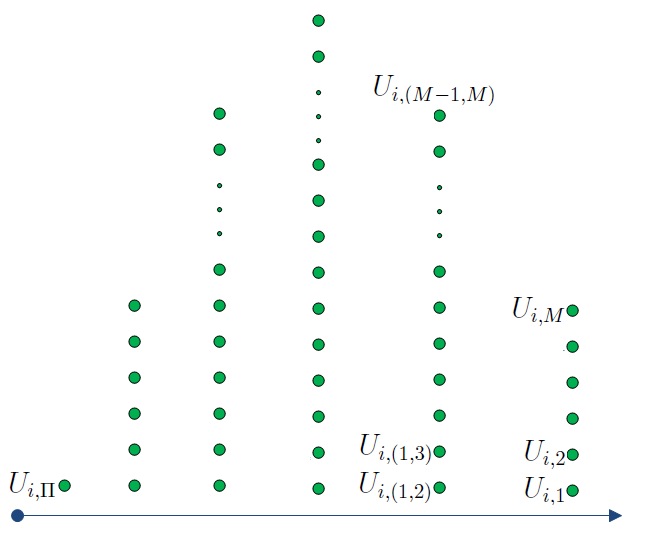}\caption{Illustrates the order of codebook generation at source $i$\label{fig:order}}

\end{figure}

\textit{Encoding} : Suppose we are given $\{U_{\Sigma}\}{}_{2^{\Pi}-\phi}$
satisfying (\ref{eq:thm_Markov}). As in section \ref{sec:2_s_2_s},
the encoding at each node is divided into 3 stages:

1) \textit{Stage 1} : We focus on the encoding at source node $E_i$. The
codebook generation is done following the order of $U_{i,\mathcal{K}},\,\,|\mathcal{K}|=M,M-1,M-2\ldots,1$
as shown in Fig. \ref{fig:order}. First, $2^{nR_{i,\Pi}^{'}}$ independent
codewords of $U_{i,\Pi}$, $u_{i,\Pi}^{n}(j)\,\, j\in\{1\ldots2^{nR_{i,\Pi}^{'}}\}$,
are generated according to the density $\prod_{t=1}^{n}P_{U_{i,\Pi}}(u_{i,\Pi}^{(t)})$.
Conditioned on each codeword $u_{i,\Pi}^{n}(j)$, $2^{nR_{i,\mathcal{K}}^{'}}$
codewords of $U_{i,\mathcal{K}}:|\mathcal{K}|=M-1$ are generated
independent of each other according to the conditional density $\prod_{t=1}^{n}P_{U_{i,\mathcal{K}}|U_{i,\Pi}}(u_{i,\mathcal{K}}^{(t)}|u_{i,\Pi}^{(t)})$.
Similarly, $\forall\mathcal{K}:|\mathcal{K}|<M$, $2^{nR_{i,\mathcal{K}}^{'}}$
codewords of $U_{i,\mathcal{K}}$ are independently generated conditioned
on each codeword tuple of $\{U_{i}\}_{\mathcal{I}_{|\mathcal{K}|+}(\mathcal{K})}$
according to $\prod_{t=1}^{n}P_{U_{i,\mathcal{K}}|\{U_{i}\}_{\mathcal{I}_{|\mathcal{K}|+}(\mathcal{K})}}(u_{i,\mathcal{K}}^{(t)}|\{u_{i}\}_{\mathcal{I}_{|\mathcal{K}|+}(\mathcal{K})}^{(t)})$.
Note that to generate the codewords of $U_{i,\mathcal{K}}$, we first need all the codebooks of $\{U_{i}\}_{\mathcal{I}_{|\mathcal{K}|+}(\mathcal{K})}$.
On observing a sequence, $x_{i}^{n}$, the encoder at $E_{i}$ attempts
to find a set of codewords, one for each variable, such that they
are all jointly typical. If it fails to find such a set, it declares
an error. Codebooks are generated similarly at all the source nodes. Note
that all the random variables $U_{i,i}\forall i\in\Sigma$ can be
set to constants without changing the rate region of Theorem \ref{thm:main}.
However, we continue to use them to avoid more complex notation.

2) \textit{Stage 2} : In stage 2, the codewords in each codebook are
divided into uniform bins. Specifically, the $2^{nR_{i,\mathcal{K}}^{'}}$
codewords in any codebook of $U_{i,\mathcal{K}}$ are subdivided into
$2^{nR_{i,\mathcal{K}}^{''}}$ bins, with each bin containing $2^{n(R_{i,\mathcal{K}}^{'}-R_{i,\mathcal{K}}^{''})}$
codewords. All the codewords which have the same bin index $m$ are
said to fall in the bin $\mathcal{C}_{i,\mathcal{K}}(m)$ $\forall m\in(1\ldots2^{nR_{i,\mathcal{K}}^{''}})$.
If in stage 1, the encoder succeeds in finding a jointly typical set
of codewords, the bin index of the codeword of $U_{i,\mathcal{K}}$
is sent as part of packet $\mathcal{P}_{i,\mathcal{K}}$.

3) \textit{Stage 3} : \textit{Power Binning} : In this stage, each
typical sequence of $X_{i}$ is assigned $2^{|\Pi_{i}|}-1$ indices,
randomly generated using uniform pmfs over $(1,\ldots,2^{\tilde{R}_{i,\mathcal{K}}})\,\,\forall\mathcal{K}\in2^{\Pi_{i}}-\phi$, 
respectively. All the sequences of $i$ which have the same bin index
$l$ are said to fall in the bin $\mathcal{B}_{i,\mathcal{K}}(l)$
$\forall l\in(1\ldots2^{n\tilde{R}_{i,\mathcal{K}}})$. On observing
a sequence $x_{i}^{n}$, if it is typical, the encoder sends the corresponding
bin indices in the packets $\mathcal{P}_{i,\mathcal{K}}:\mathcal{K}\in2^{\Pi_{i}}-\phi$,
in addition to the bin indices in stage 2. If it is not typical, the
encoder declares an error. Note that all packets from source node $E_i$
to a subset of sinks $\mathcal{K}$ such that $\mathcal{K}\subseteq2^{\Pi_{i}}-\phi$,
carry two bin indices, one each from stages 2 and 3, respectively.

In Appendix \ref{app:Proof-of-Theorem}, we show that, if the rates
$R_{i,\mathcal{K}}^{'}$ satisfy (\ref{eq:thm_0}), then the probability
of encoding error asymptotically approaches zero, i.e., we can, with
probability approaching 1, find a codeword tuple, one from each codebook
such that all the codewords are jointly typical if the rates satisfy
(\ref{eq:thm_0}). Let the codewords, which are jointly typical with
$x_{i}^{n}$, be denoted as $u_{i,\mathcal{K}}^{*}$ $\forall\mathcal{K}\in\mathcal{J}(\Pi)=2^{\Pi}-\phi$.
To ensure joint typicality of $(\{x\}_{\Sigma}^{n},\{u_{\Sigma}^{*}\}_{\mathcal{J}(k)})$,
we require a stronger version of the {}``conditional Markov lemma''
in \cite{Wagner}. We state and prove this stronger version, called the ``conditional Markov lemma for mutual covering'' in Appendix B. From this lemma, it follows that $\left(\{x\}_{\Sigma}^{n},\{u_{\Sigma}^{*}\}_{\mathcal{J}(k)}\right)\in\mathcal{T}_{\epsilon}^{n}\left(\{X\}_{\Sigma}^{n},\{U_{\Sigma}^{*}\}_{\mathcal{J}(k)}\right)$
with very high probability given that the encoding at all the source nodes is
error free. Let the bin indices of $u_{i,\mathcal{K}}^{*}$ (assigned
in stage 2) be denoted by $m_{i,\mathcal{K}}\,\,\forall\mathcal{K}\in2^{\Pi}-\phi$
and let the bin indices of $x_{i}^{n}$ (assigned in stage 3) be denoted
by $l_{i,\mathcal{K}}\forall\mathcal{K}\in2^{\Pi_{i}}-\phi$.

\textit{Decoding} : We focus on a particular sink $S_{k}$. Sink $S_{k}$
receives all the indices $\{m_{\Sigma}\}_{\mathcal{J}(k)}$ of stage
2 of encoding from all source nodes. It also receives $\{l_{\Sigma_{k}}\}_{\mathcal{J}(k)}$
of stage 3 of encoding from source nodes $\Sigma_{k}$. In the first stage
of decoding, it begins decoding $u_{i,\mathcal{J}(k)}^{*}\,\,\forall i\in\Sigma$
by looking for a unique jointly typical codeword tuple from $\{\mathcal{C}_{i,\mathcal{J}(k)}(m_{i,\mathcal{J}(k)});\forall i\in\Sigma\}$.
Clearly $\{u_{\Sigma}^{*}\}_{\mathcal{J}(k)}$ satisfies this property.
If the decoder finds another such jointly typical codeword tuple in
the received bins, it declares an error. In Appendix \ref{app:Proof-of-Theorem}, we show that if conditions (\ref{eq:thm_1})  are satisfied by  $R_{i,\mathcal{K}}^{''}$, then the probability that the decoder finds another such jointly typical codeword tuple approaches zero.

In the last stage of decoding, after having decoded all $\{u_{\Sigma}^{*}\}_{\mathcal{J}(k)}$,
the decoder looks for unique source sequences from $\bigcap\{\mathcal{B}_{i,\mathcal{K}}(l_{i,\mathcal{K}}):i\in\Sigma_{k},\mathcal{K}\ni k\}$
which are jointly typical with $\{u_{\Sigma}^{*}\}_{\mathcal{J}(k)}$.
Hence what remains is to find conditions on $\tilde{R}_{i,\mathcal{K}}$
to ensure lossless reconstruction of the respective sources at each
sink. Following similar steps as in \cite{Han_Kobayashi,Slepian_Wolf},
it is easy to show that this probability can be made arbitrarily small
if (\ref{eq:thm_2}) is satisfied $\forall\Gamma\in2^{\Sigma_{k}}-\phi$.
We have shown that if the rates satisfy the conditions in Theorem
\ref{thm:main}, the probability of decoding error at each sink can be
made arbitrarily small. Arbitrarily small decoding error ensures that the decoder decodes the correct sequence with very
high probability. Hence, if the rate constraints are satisfied, for
any $\epsilon>0$, we can find a sufficiently large $n$ such that:\begin{equation}
P(\hat{X}_{\sum_{j}}^{n}\neq X_{\sum_{j}}^{n})<\epsilon\end{equation}

Recall that packets from source node $E_{i}$ to sinks $\mathcal{K}\subseteq\Pi_{i}$
carry both $m_{i,\mathcal{K}}$ (at rate $R_{i,\mathcal{K}}^{''}$)
and $l_{i,\mathcal{K}}$ (at rate $\tilde{R}_{i,\mathcal{K}}$). While
the other packets carry only $m_{i,\mathcal{K}}$ (at rate $R_{i,\mathcal{K}}^{''}$).
Hence, the rates of each packet must satisfy the following constraints
for lossless decoding of the requested sources:\begin{equation}
R_{i,\mathcal{K}}\geq\begin{cases}
R_{i,\mathcal{K}}^{''}+\tilde{R}_{i,\mathcal{K}} & \mbox{if }\mathcal{K}\subseteq2^{\Pi_{i}}-\phi\\
R_{i,\mathcal{K}}^{''} & \mbox{if }\mathcal{K}\nsubseteq2^{\Pi_{i}}-\phi\end{cases}\label{eq:thm_3-1}\end{equation}
proving the theorem.\end{proof}
\begin{rem}
\textbf{A note on separability of distributed compression and routing}
: It was shown in \cite{Networked_Slepian_Wolf} that the two problems
of DSC (Slepian-Wolf compression) and optimum \textit{broadcast routing}
are separable problems, i.e., the optimum routes can be found without
the knowledge of the achievable rates, and vice versa, the rate
region can be found without the knowledge of the routes. However,
we demonstrated in \cite{ISIT10} that such separability holds only
under the `no helpers' assumption. We also showed that the extent
of suboptimality due to separating DSC and broadcast routing is substantial and potentially unbounded 
when helpers are allowed to communicate. In general the optimum rate
region cannot be found without the knowledge of the network costs
for broadcast routing. However, for DIR, the two problems of finding
the optimum rate region for the tuple $(R_{i,\mathcal{K}}\,\,\forall i\in\Sigma,\mathcal{K}\in2^{\Pi}-\phi)$
and finding the optimum routes from the source nodes to the sinks can be
separated and dealt independently, without entailing any loss of optimality.
Note that even though DIR has the inherent advantage of separability,
finding the optimum operating point requires optimizing over an $N\times2^{M}$
dimensional space and the effective complexity remains the same as that
for broadcast routing.
\end{rem}

\section{Outerbounds to certain special scenarios\label{sec:Outerbounds-to-certain}}

We note that the converse to the achievability region does not hold
in general. However, we can prove the converse for two important special cases.

\subsection{When there are no helpers\label{sub:outer_no_helper}}
\begin{thm}
When each sink is allowed to receive packets only from sources it intends to reconstruct, the complete rate region for dispersive
information routing is given by: $\forall j\in\Pi$ and $\forall\mathcal{S}\in2^{\Sigma_{j}}-\phi$:
\begin{equation}
\sum_{i\in\mathcal{S}}\sum_{\mathcal{K}\in2^{\Pi_{i}}-\phi,\,\,\mathcal{K}\ni j}R_{i,\mathcal{K}}\geq H\left(\{X\}_{\mathcal{S}}|\{X\}_{\Sigma_{j}-\mathcal{S}}\right)\end{equation}
 It is achieved by `Power Binning'.
\end{thm}
\begin{proof} In the achievable rate region of Theorem \ref{thm:main},
setting $U_{i,\mathcal{S}}=\Phi\,\,\forall i\in\Sigma,\mathcal{S}\in2^{\Pi}-\phi$,
where $\Phi$ is a constant, leads to the above rate region. The converse
to this rate region follows directly from the converse to the lossless source
coding theorem \cite{Cover-book}. We omit the proof as it is straightforward.
\end{proof}

\subsection{A 2-Sink network with a single helper\label{sub:outer_1_helper}}

\begin{figure}
\centering\includegraphics[scale=0.3]{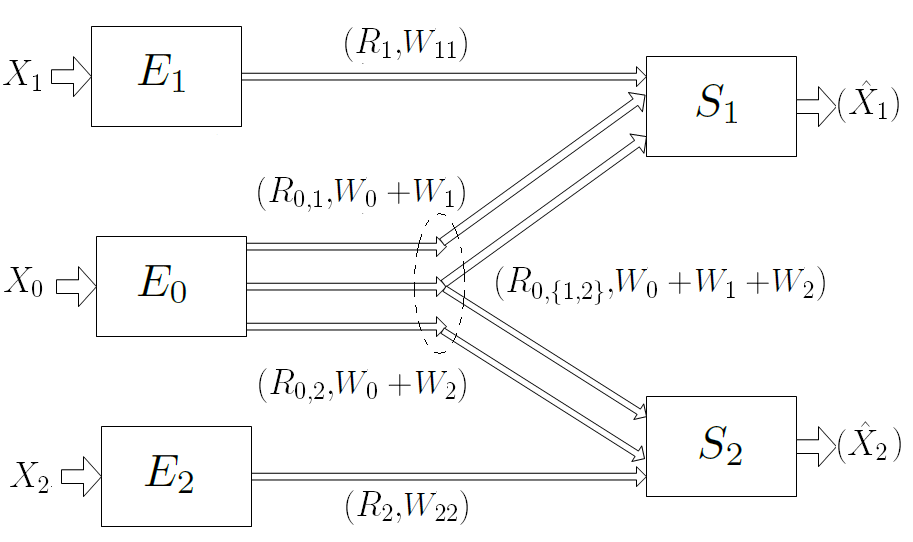}\caption{Example of a 2-sink, 1 Helper DIR\label{fig:2_sink_1_help}}

\end{figure}

The converse can be proven in general for any 2 sink network with
a single helper. However, to avoid complex notation, we just give
a simple example of a 2 sink network with a single helper and prove
the converse to the rate region. The proof of converse for a general
2 sink network with a single helper follows in similar lines.

Consider the network shown in Fig. \ref{fig:2_sink_1_help}, with
3 source nodes and 2 sinks. The three source nodes $E_{1},E_{0},E_{2}$ observe
three correlated memoryless random sequences $X_{1}^n,X_{0}^n,X_{2}^n$, respectively.
The two sinks $S_{1}$ and $S_{2}$ respectively reconstruct $X_{1}^n$
and $X_{2}^n$ losslessly. Note that $E_{0}$ acts as a helper to both
the sinks. Our objective is to find the rate region for the tuple
$(R_{1},R_{2},R_{0,1},R_{0,2},R_{0,\{1,2\}})$ for lossless reconstruction
of the respective sources. It is important to remember that our ultimate
objective is to find the minimum communication cost, which follows
by finding the point in the rate region that minimizes the following cost function:\begin{eqnarray}
C_{DIR}^{*}=W_{11}R_{1}+W_{22}R_{2}+(W_{0}+W_{1})R_{0,1}\nonumber \\
+(W_{0}+W_{2})R_{0,2}+(W_{0}+W_{1}+W_{2})R_{0,\{1,2\}}\end{eqnarray}The following theorem establishes the
complete rate region.
\begin{thm}
\label{thm:helper}Let $(U_{0},U_{1},U_{2})$ be random variables
distributed over arbitrary finite sets $\mathcal{U}_{0}\times\mathcal{U}_{1}\times\mathcal{U}_{2}$,
jointly distributed with $(X_{1},X_{0},X_{2})$ such that the following
hold: \begin{eqnarray}
X_{1} & \leftrightarrow & X_{0}\leftrightarrow(U_{0},U_{1},U_{2})\nonumber \\
X_{2} & \leftrightarrow & X_{0}\leftrightarrow(U_{0},U_{1},U_{2})\label{eq:2_sink_1_helper_markov}\end{eqnarray}
Then any rate tuple satisfying the following constraints is achievable
for the 2-Sink 1-Helper DIR problem:\begin{eqnarray}
R_{0,12} & \geq & I(X_{0};U_{0})\nonumber \\
R_{0,1} & \geq & I(X_{0};U_{1}|U_{0})\nonumber \\
R_{0,2} & \geq & I(X_{0};U_{2}|U_{0})\nonumber \\
R_{1,1} & \geq & H(X_{1}|U_{0},U_{1})\nonumber \\
R_{2,2} & \geq & H(X_{2}|U_{0},U_{2})\label{eq:2_sink_1_helper_them}\end{eqnarray}
The closure of the achievable rates over all such $(U_{0},U_{1},U_{2})$
is the complete rate region for this setup. \end{thm}
\begin{proof}
\textbf{Achievability} : Let $(U_{0},U_{1},U_{2})$ be any random
variables satisfying (\ref{eq:2_sink_1_helper_markov}). The following
achievable rate region is obtained by setting $U_{0,12}=U_{0}$, $U_{0,1}=U_{1}$, $U_{0,12}=U_{2}$ and all the remaining random variables to constants in the general achievable rate region of Theorem \ref{thm:main}:\begin{eqnarray}
R_{0,12} & \geq & I(X_{0};U_{0})\nonumber \\
R_{0,12}+R_{0,1} & \geq & I(X_{0};U_{0})+I(X_{0};U_{1}|U_{0})\nonumber \\
R_{0,12}+R_{0,2} & \geq & I(X_{0};U_{0})+I(X_{0};U_{2}|U_{0})\nonumber \\
R_{0,12}+R_{0,1}+R_{0,2} & \geq & I(X_{0};U_{1},U_{2},U_{0})+I(U_{1};U_{2}|U_{0})\nonumber \\
R_{1,1} & \geq & H(X_{1}|U_{0},U_{1})\nonumber \\
R_{2,2} & \geq & H(X_{2}|U_{0},U_{2})\label{eq:2_sink_1_helper_them-1}\end{eqnarray}
We further restrict the joint density to satisfy the following Markov condition in addition to  (\ref{eq:2_sink_1_helper_markov}):\begin{equation}
U_{1}\leftrightarrow(X_{0},U_{0})\leftrightarrow U_{2}\label{eq:pairwise_markov}\end{equation}
On using this Markov condition in (\ref{eq:2_sink_1_helper_them-1}),
the sum rate constraint on $R_{0,12}+R_{0,1}+R_{0,2}$ becomes:\begin{eqnarray}
R_{0,12}+R_{0,1}+R_{0,2}\geq I(X_{0};U_{0})+I(X_{0};U_{1}|U_{0})& & \nonumber \\ 
+I(X_{0};U_{2}|U_{0}) & & \label{eq:redundant}\end{eqnarray}
Observe that if a rate tuple satisfies (\ref{eq:2_sink_1_helper_them}), then it also satisfies  (\ref{eq:2_sink_1_helper_them-1}) and hence the region given by (\ref{eq:2_sink_1_helper_them}) is achievable for the 2-Sink 1-Helper problem shown in Fig. \ref{fig:2_sink_1_help}.

\textbf{Converse} : Recall the notation in the definition of an achievable
rate region in Section \ref{sub:general_ach_region}. The output of
encoder 1 is denoted $f_{1}^{E}(X_{1}^{n})=T_{1}$ and the output
of encoder 2 is $f_{2}^{E}(X_{2}^{n})=T_{2}$. Remember that $0\leq T_{1}\leq 2^{nR_{1}}$
and $0\leq T_{1}\leq 2^{nR_{2}}$. Similarly the encoder at $E_{0}$ transmits
3 indices denoted by $(T_{0,1},T_{0,2},T_{0,12})$ which are routed
to the respective sinks. Sink $S_{1}$ receives $(T_{1},T_{0,1},T_{0,12})$
and reconstructs $X_{1}^{n}$ with vanishing probability of error.
Similarly sink $S_{2}$ receives $(T_{2},T_{0,2},T_{0,12})$ and reconstructs
$X_{2}^{n}$ losslessly. We need to prove that for any code with vanishing
probability of error, the rates must satisfy (\ref{eq:2_sink_1_helper_them})
for some $(U_{0},U_{1},U_{2})$ satisfying (\ref{eq:2_sink_1_helper_markov}).

We follow standard converse techniques to prove the above claim. We
begin with the following series of inequalities:\begin{eqnarray}
nR_{0,12} & \geq & H(T_{0,12})\geq I(X_{0}^{n};T_{0,12})\nonumber \\
 & = & \sum_{i=1}^{n}I(X_{0}^{i};T_{0,12}|X_{0}^{1,i-1})\nonumber \\
 & =^{(a)} & \sum_{i=1}^{n}I(X_{0}^{i};T_{0,12},X_{0}^{1,i-1})\nonumber \\
 & =^{(b)} & \sum_{i=1}^{n}I(X_{0}^{i};U_{0,12}^{i})\label{eq:1h_conv_1}\end{eqnarray}
 where $(a)$ follows from the memoryless property of the sources
and $(b)$ follows by setting $U_{0,12}^{i}=(T_{0,12},X_{0}^{1,i-1})$.
Here $X_{0}^{i}$ denotes the $i$'th realization of $X_{0}^n$ and
$X_{0}^{1,i-1}$ denotes the first $i-1$ realizations of $X_{0}^n$.
Next we have:\begin{eqnarray}
nR_{0,1} & \geq & H(T_{0,1})\geq H(T_{0,1}|T_{0,12})\nonumber \\
 & \geq & I(X_{0}^{n};T_{0,1}|T_{0,12})\nonumber \\
 & = & \sum_{i=1}^{n}I(X_{0}^{i};T_{0,1}|T_{0,12},X_{0}^{1,i-1})\nonumber \\
 & = & \sum_{i=1}^{n}I(X_{0}^{i};U_{0,1}^{i}|U_{0,12}^{i})\label{eq:1h_conv_2}\end{eqnarray}
 Where $U_{0,1}^{i}=(T_{0,1})\,\,\forall i$. Similarly, we can show
that:\begin{equation}
nR_{0,2}\geq\sum_{i=1}^{n}I(X_{0}^{i};U_{0,2}^{i}|U_{0,12}^{i})\label{eq:1h_conv_3}\end{equation}
 where $U_{0,2}^{i}=(T_{0,2})\,\,\forall i$. Note that as $(T_{0,1},T_{0,2},T_{0,12},X_{0}^{1,i-1})$
depends on $(X_{1}^{i},X_{2}^{i})$ only through $X_{0}^{i}$, we
have the following two Markov chain conditions:\begin{eqnarray}
X_{1}^{i} & \leftrightarrow & X_{0}^{i}\leftrightarrow(U_{0}^{i},U_{1}^{i},U_{2}^{i})\nonumber \\
X_{2}^{i} & \leftrightarrow & X_{0}^{i}\leftrightarrow(U_{0}^{i},U_{1}^{i},U_{2}^{i})\label{eq:1h_conv_4}\end{eqnarray}
 Further, we need lossless reconstruction of $X_{1}^{n}$ at $S_{1}$. The following series of inequalities hold:\begin{eqnarray}
nR_{1} & \geq & H(T_{1})\nonumber \\
 & \geq & H(T_{1}|T_{0,12},T_{0,1})\nonumber \\
 & = & H(T_{1}|T_{0,12},T_{0,1})+H(X_{1}^{n}|T_{0,12},T_{0,1},T_{1})\nonumber \\
 &  & -H(X_{1}^{n}|T_{0,12},T_{0,1},T_{1})\nonumber \\
 & \geq^{(a)} & H(X_{1}^{n},T_{1}|T_{0,12},T_{0,1})-n\epsilon_{n}\nonumber \\
 & = & H(X_{1}^{n}|T_{0,12},T_{0,1})-n\epsilon_{n}\nonumber \\
 & = & \sum_{i=1}^{n}H(X_{1}^{i}|X_{1}^{i-1},T_{0,12},T_{0,1})-n\epsilon_{n}\nonumber \\
 & = & \sum_{i=1}^{n}H(X_{1}^{i}|U_{0,12}^{i},U_{0,1}^{i})-n\epsilon_{n}\label{eq:1h_conv_5}\end{eqnarray}
 where $(a)$ follows from Fano's inequality, i.e., $H(X_{1}^{n}|T_{1},T_{0,1},T_{0,12})<n\epsilon_{n}$. Similarly, for lossless
reconstruction at $S_{2}$, we have:\begin{equation}
nR_{2}\geq\sum_{i=1}^{n}H(X_{2}^{i}|U_{0,12}^{i},U_{0,2}^{i})-n\epsilon_{n}\label{eq:1h_conv_6}\end{equation}
 We next introduce a time sharing random variable $Q\sim\mbox{Unif}[1:n]$,
independent of $(X_{0}^{n},X_{1}^{n},X_{2}^{n},U_{0,1}^{n},U_{0,2}^{n},U_{0,12}^{n})$,
so that we can rewrite (\ref{eq:1h_conv_1}), (\ref{eq:1h_conv_2}),
(\ref{eq:1h_conv_3}), (\ref{eq:1h_conv_5}) and (\ref{eq:1h_conv_6})
as:\begin{eqnarray}
nR_{0,12} & \geq & I(X_{0}^{Q};U_{0,12}^{Q}|Q)=I(X_{0}^{Q};U_{0,12}^{Q},Q)\nonumber \\
nR_{0,1} & \geq & I(X_{0}^{Q};U_{0,1}^{Q}|U_{0,12}^{Q},Q)\nonumber \\
 & = & I(X_{0}^{Q};U_{0,1}^{Q},Q|U_{0,12}^{Q},Q)\nonumber \\
nR_{0,2} & \geq & I(X_{0}^{Q};U_{0,2}^{Q}|U_{0,12}^{Q},Q)\nonumber \\
 & = & I(X_{0}^{Q};U_{0,2}^{Q},Q|U_{0,12}^{Q},Q)\nonumber \\
nR_{1} & \geq & H(X_{1}^{Q}|U_{0,12}^{Q},U_{0,1}^{Q},Q)\nonumber \\
nR_{2} & \geq & H(X_{2}^{Q}|U_{0,12}^{Q},U_{0,2}^{Q},Q)\label{eq:1h_conv_7}\end{eqnarray}
 Setting $(U_{0,12}^{Q},Q)=U_{0,12}$, $(U_{0,1}^{Q},Q)=U_{0,1}$,
$(U_{0,2}^{Q},Q)=U_{0,2}$ and observing that $(X_{0}^{Q},X_{1}^{Q},X_{2}^{Q})$
has the same density as $(X_{0},X_{1},X_{2})$ we get the rate region
given in (\ref{eq:2_sink_1_helper_them}).
\end{proof}
\textbf{Example to demonstrate strict improvement:} Next we show that
DIR achieves strictly lower communication cost for the single helper
network shown in Fig. \ref{fig:2_sink_1_help}. This example demonstrates
the freedom DIR provides over broadcast routing by sending only the
relevant information to each sink, even when the information is from
a helper. The complete rate region under broadcast routing for the
example shown in Fig. \ref{fig:2_sink_1_help} was determined in \cite{Wyner_Source_coding,Ahlswede_side_info}
and is given by the closure of the following rate tuples over
all random variables $U_{0}$ satisfying $(X_{1},X_{2})\leftrightarrow X_{0}\leftrightarrow U_{0}$:\begin{eqnarray}
R_{0,12} & \geq & I(X_{0};U_{0})\nonumber \\
R_{1,1} & \geq & H(X_{1}|U_{0})\nonumber \\
R_{2,2} & \geq & H(X_{2}|U_{0})\label{eq:2_sink_1_helper_them-2}\end{eqnarray}

We consider the example where $(X_{0},X_{1},X_{2})$ are binary
symmetric sources such that $X_{1}\leftrightarrow X_{0}\leftrightarrow X_{2}$
holds. The transition probabilities are such that $X_{1}$ and $X_{2}$
are obtained as outputs of two independent binary symmetric channels
with $X_{0}$ as input and cross-over probabilities of $P_{1}$ and
$P_{2}$, respectively. Let us say that the network costs are such
that $E_{1}$ and $E_{2}$ send at rates $\Delta$ more than their
respective conditional entropies (for some $\Delta>0$), i.e., $R_{1}=H_{b}(P_{1})+\Delta$
and $R_{2}=H_{b}(P_{2})+\Delta$ where $H_{b}(\cdot)$
denotes the binary entropy function (note that the conditional entropy
is the minimum information each encoder has to send). Wyner \cite{Wyner_Source_coding}
(see also \cite{Effros}) showed that the minimum rate from $E_{0}$
to the two sinks under broadcast routing is given by:\begin{align}
R_{0} & \geq\max_{P_{0}\in\{P_{01},P_{02}\}}1-H_{b}(P_{0})\label{eq:example_br}\end{align}
where $P_{01}$ and $P_{02}$ solve the respective equations $H_{b}(P_{1}\bullet P_{01})=H_{b}(P_{1})+\Delta$
and $H_{b}(P_{2}\bullet P_{02})=H_{b}(P_{2})+\Delta$ where $P_{1}\bullet P_{2}=P_{1}P_{2}+(1-P_{1})(1-P_{2})$.
The optimum $U_{0}$ which achieves the boundary points is obtained
by passing $X_{0}$ through a binary symmetric channel (BSC) with
cross over probability $P_{0}$. Again observe that, if the sinks
$S_{1}$ and $S_{2}$ receive information from $E_{1}$ and $E_{2}$
at rates $H_{b}(P_{1})+\Delta$ and $H_{b}(P_{2})+\Delta$, they require
information from $E_{0}$ at rates $1-H_{b}(P_{01})$ and $1-H_{b}(P_{02})$, 
respectively. However, broadcast routing sends information at the maximum
of the two to both sinks and hence if $P_{1}\neq P_{2}$ (which
in turn implies $P_{01}\neq P_{02}$ in general), there is sub-optimality
on either one of the two branches connecting from the collector to
the two sinks. 

On the other hand, using DIR, we can achieve minimum
rates on all the branches. To prove this claim, without loss of generality, 
let us assume that $0.5>P_{01}>P_{02}>0$. Consider the following
joint density for $(U_{0},U_{1},U_{2})$ in Theorem \ref{thm:helper}.
$U_{2}$ is the output when $X_{0}$ is sent through a BSC with cross
over probability $P_{02}$ and $U_{0}$ is the output when $U_{2}$
is sent through a BSC with cross over probability $P_{012}$ where
$P_{02}\bullet P_{012}=P_{01}$. $U_{1}$ is set as a constant. It
is easy to verify from Theorem \ref{thm:helper} that the following rates are achievable:\begin{align}
R_{0,12} & =1-H_{b}(P_{01})\nonumber \\
R_{0,2} & =H_{b}(P_{01})-H_{b}(P_{02})\label{eq:example_DIR}\end{align}
which implies that the two sinks receive at their respective minima
leading to the conclusion that DIR achieves the minimum communication cost for this example.

\section{Conclusion\label{sec:Conclusion}}

This paper considers a new routing paradigm called dispersive information
routing, wherein each intermediate node is allowed to ``split a packet''
and forward subsets of the information on individual forward paths. We demonstrated
using simple examples the gains of DIR over broadcast routing. Unlike
network coding, this new routing technique can be realized using conventional
routers with source nodes transmitting multiple smaller packets into the
network. This paradigm introduces a new class of information theoretic
problems. We derived an achievable rate region for this setup using
principles from multiple descriptions encoding and Han and Kobayashi
decoding which is complete for certain special cases of the setup.

%\begin{thebibliography}{34}
\bibliographystyle{IEEEtran}
\bibliography{DIR}
%\end{thebibliography}

\appendix
%dummy comment inserted by tex2lyx to ensure that this paragraph is not empty

\section*{Appendix A: Bounding Encoding/Decoding Errors in Theorem \ref{thm:main}
\label{app:Proof-of-Theorem}}

\begin{proof} 

\textit{Probability of encoding error} : Let us analyze the probability
of encoding error at source node $E_i$. Let $\mathcal{E}$ denote the event of an encoding
error. We have:\begin{eqnarray}
P(\mathcal{E}) & = & P(\mathcal{E}|x_{i}^{n}\in\mathcal{T}_{\epsilon}^{n})P(x_{i}^{n}\in\mathcal{T}_{\epsilon}^{n})\nonumber \\
 &  & +P(\mathcal{E}|x_{i}^{n}\notin\mathcal{T}_{\epsilon}^{n})P(x_{i}^{n}\notin\mathcal{T}_{\epsilon}^{n})\end{eqnarray}
 From standard typicality arguments, we have $P(x_{i}^{n}\notin\mathcal{T}_{\epsilon}^{n})\rightarrow0$
as $n\rightarrow\infty$. Hence, it is sufficient to find conditions
on the rates to bound $P(\mathcal{E}|x_{i}^{n}\in\mathcal{T}_{\epsilon}^{n})$. 

Towards finding conditions on the rate to bound $P(\mathcal{E}|x_{i}^{n}\in\mathcal{T}_{\epsilon}^{n})$,
we define the random variables :
\begin{equation}
\chi(\{j\}_{\mathcal{J}(\Pi)})=\begin{cases}
1 & \mbox{if}\,\,\left(x_{i}^{n},u_{i}^{n}(\{j\}_{\mathcal{J}(\Pi)})\right)\in\mathcal{T}_{\epsilon}^{n}\\
0 & \mbox{else}\end{cases}\label{eq:xi_defn}\end{equation}We have $P(\mathcal{E}|x_{i}^{n}\in\mathcal{T}_{\epsilon}^{n})=P(\Psi=0)$
where $\Psi=\sum_{\mathcal{J}(\Pi)}\chi(\{j\}_{\mathcal{J}(\Pi)})$.
From Chebyshev's inequality, it follows that:\begin{equation}
P(\Psi=0)\leq P\left[|\Psi-E[\Psi]|\geq E[\Psi]/2\right]\leq\frac{4\mbox{Var}(\Psi)}{\left(E[\Psi]\right)^{2}}\label{eq:P_3}\end{equation}
 From Lemma 3.1 in \cite{Han_Kobayashi}, we can bound $E[\Psi]$
as follows:\begin{equation}
E[\Psi]\geq2^{n\sum_{\mathcal{K}\in\mathcal{J}(\Pi)}R_{\mathcal{K}}^{'}-n\left(\alpha(i,\mathcal{J}(\Pi))+\epsilon\right)}\label{eq:E_xi}\end{equation}
 where\begin{eqnarray}
\alpha(i,\mathcal{Q}) & = & -H\left(\{U_{i}\}_{\mathcal{Q}}|X_{i}\right)\nonumber \\
 &  & +\sum_{\mathcal{K}\in\mathcal{Q}}H\left(U_{i,\mathcal{K}}|\{U_{i}\}_{\mathcal{I}_{|\mathcal{K}|+}(\mathcal{K})}\right)\label{eq:alpha_defn}\end{eqnarray}
 $\forall i,\mathcal{Q}\subseteq\mathcal{J}(\Pi)$. We follow the
convention $\alpha_{W}(i,\phi)=0$. Next consider $\mbox{Var}(\Psi)=E[\Psi^{2}]-\left(E[\Psi]\right)^{2}$
where, \begin{eqnarray}
E[\Psi^{2}]=\sum_{\{j\}_{\mathcal{J}(\Pi)}}\sum_{\{k\}_{\mathcal{J}(\Pi)}}E\left[\chi(\{j\}_{\mathcal{J}(\Pi)})\chi(\{k\}_{\mathcal{J}(\Pi)})\right]\nonumber \\
=\sum_{\{j\}_{\mathcal{J}(\Pi)}}\sum_{\{k\}_{\mathcal{J}(\Pi)}}P\left[\chi(\{j\}_{\mathcal{J}(\Pi)})=1,\chi(\{k\}_{\mathcal{J}(\Pi)})=1\right]\label{eq:E_xi_s}\end{eqnarray}
The probability in (\ref{eq:E_xi_s}) depends on whether $u_{i}^{n}(\{j\}_{\mathcal{J}(\Pi)})$
and $u_{i}^{n}(\{k\}_{\mathcal{J}(\Pi)})$ are equal for a subset
of indices. Let $\mathcal{Q}\subseteq\mathcal{J}(\Pi),\,\,\mathcal{Q}\neq\phi$,
such that $\{j\}_{\mathcal{Q}}=\{k\}_{\mathcal{Q}}$. Observe that,
due to the hierarchical structure in the conditional codebook generation
mechanism, for $u_{i}^{n}(\{j\}_{\mathcal{Q}})=u_{i}^{n}(\{k\}_{\mathcal{Q}})$
to hold, $\mathcal{Q}$ must be such that,\begin{equation}
\mbox{if }\mathcal{K}\in\mathcal{Q}\,\,\Rightarrow\,\,\,\mathcal{I}_{|\mathcal{K}|+}(\mathcal{K})\subset\mathcal{Q}\label{eq:cond_Q}\end{equation}
i.e., $\mathcal{Q}\in\mathcal{Q}^{*}$ given in (\ref{eq:cond_Q_main}).
It follows from the codebook generation mechanism that given the codeword
tuple $\{u_{i}^{n}(\{j\}_{\mathcal{Q}})\}$, tuples $\{u_{i}^{n}(\{j\}_{\mathcal{J}(\Pi)-\mathcal{Q}})\}$
and $\{u_{i}^{n}(\{k\}_{\mathcal{J}(\Pi)-\mathcal{Q}})\}$ are independent
and identically distributed. Hence we can rewrite the probability
in (\ref{eq:E_xi_s}) for some $\mathcal{Q}\subseteq\mathcal{J}(\Pi),\,\,\mathcal{Q}\neq\phi$,
as:\begin{eqnarray}
P\left[\mathcal{E}(\{j\}_{\mathcal{J}(\Pi)})\cap\mathcal{E}(\{k\}_{\mathcal{J}(\Pi)})\right]=\left(\frac{P\left[\mathcal{E}(\{j\}_{\mathcal{J}(\Pi)})\right]}{P\left[\mathcal{E}(\{j\}_{\mathcal{Q}})\right]}\right)^{2}\nonumber \\
\times P\left[\mathcal{E}(\{j\}_{\mathcal{Q}})\right]\label{eq:upper_bound_p_var}\end{eqnarray}
However, note that if $\mathcal{Q}=\phi$, then:\begin{eqnarray} P\left[\mathcal{E}(\{j\}_{\mathcal{J}(\Pi)})\cap\mathcal{E}(\{k\}_{\mathcal{J}(\Pi)})\right]=\left(P\left[\mathcal{E}(\{j\}_{\mathcal{J}(\Pi)})\right]\right)^{2}\end{eqnarray}
Next, the total number of ways of choosing $\{j\}_{\mathcal{J}(\Pi)}$
and $\{k\}_{\mathcal{J}(\Pi)}$ such that they overlap in the subset
$\mathcal{Q}$ is: \begin{eqnarray}
2^{n\sum_{\mathcal{K}\in\mathcal{Q}}R_{i,\mathcal{K}}^{'}}\prod_{\mathcal{K}\in\mathcal{J}(\Pi)-\mathcal{Q}}2^{nR_{i,\mathcal{K}}^{'}}(2^{nR_{i,\mathcal{K}}^{'}}-1)\nonumber \\
\leq2^{n\{\sum_{\mathcal{K}\in\mathcal{J}(\Pi)}R_{i,\mathcal{K}}^{'}+2\sum_{\mathcal{K}\in\mathcal{J}(\Pi)-\mathcal{Q}}R_{i,\mathcal{K}}^{'}\}}\label{eq:upper_bound_rateE_var}\end{eqnarray}
On substituting (\ref{eq:upper_bound_p_var}) and (\ref{eq:upper_bound_rateE_var})
in (\ref{eq:E_xi_s}), we bound $\mbox{Var}(\Psi)$ as:

\begin{eqnarray}
\mbox{Var}(\Psi) & \leq & \sum\biggl\{2^{-2n\left(\alpha(i,\mathcal{J}(\Pi))-\sum_{\mathcal{K}\in\mathcal{J}(\Pi)}R_{i,\mathcal{K}}^{'}\right)}\nonumber \\
 &  & 2^{n\left(\alpha(i,\mathcal{Q})-\sum_{\mathcal{K}\in\mathcal{Q}}R_{i,\mathcal{K}}^{'}\right)+5n\epsilon}\biggl\}\label{eq:Var_xi}\end{eqnarray}
where the summation is over all non-empty $\mathcal{Q}$ such that
(\ref{eq:cond_Q}) holds. Observe that the term corresponding to
$\mathcal{Q}=\phi$ gets canceled with the `$(\mathcal{E}[\Psi])^{2}$' term
in $\mbox{Var}(\Psi)$. Inserting, (\ref{eq:Var_xi}) and (\ref{eq:E_xi})
in (\ref{eq:P_3}), we get :\begin{equation}
P(E|x_{i}^{n}\in\mathcal{T}_{\epsilon}^{n})\leq4\sum2^{n\left(\alpha(i,\mathcal{Q})-\sum_{\mathcal{K}\in\mathcal{Q}}R_{\mathcal{K}}^{'}\right)+7n\epsilon}\label{eq:P_4}\end{equation}
where the summation is over all non-empty $\mathcal{Q}$ satisfying
(\ref{eq:cond_Q}). Hence, the probability of encoding error at all
the source nodes can be made arbitrarily small if:

\begin{equation}
\sum_{\mathcal{K}\in\mathcal{Q}}R_{i,\mathcal{K}}^{'}\geq\alpha(i,\mathcal{Q})+7\epsilon\label{eq:proof_1}\end{equation}
$\forall i,\mathcal{Q}$ satisfying (\ref{eq:cond_Q}).

\textit{Probability of decoding error} : We focus on decoding at sink
$S_{k}$. We first bound the probability of error for the first stage
of decoding. The decoder looks for a unique codeword tuple from $\Bigl\{\{\mathcal{C}_{\Sigma}\}_{\mathcal{J}(k)}\bigl(\{m_{\Sigma}\}_{\mathcal{J}(k)}\bigr)\Bigr\}$
which are jointly typical. We know that $\{u_{\Sigma}^{*}\}_{\mathcal{J}(k)}$
are jointly typical from the Markov Lemma in Appendix B.
We have to find conditions on $R_{i,\mathcal{S}}^{''}$ to ensure
no other tuple satisfies this property. Denote by $\mathcal{F}$ the
event of a decoding error given the encoding is error-free. Due to
the symmetry in codebook generation, we can assume that the index
tuple of $\{u_{\Sigma}^{*}\}_{\mathcal{J}(k)}$ is $(1,\ldots,1)$.
Let $\{j_{\Sigma}\}_{\mathcal{J}(k)}$ be an index tuple such that:\begin{equation}
\{j_{\Sigma}\}_{\mathcal{J}(k)}\neq(1,\ldots,1)\end{equation}
Define the event $\mathcal{F}(\{j_{\Sigma}\}_{\mathcal{J}(k)})$ as:\begin{equation}
\mathcal{F}(\{j_{\Sigma}\}_{\mathcal{J}(k)})=\Bigl\{\left(u_{\Sigma}^{n}(\{j_{\Sigma}\}_{\mathcal{J}(k)})\right)\in\mathcal{T}_{\epsilon}^{n}\Bigr\}\end{equation}
It then follows from union bound that:\begin{eqnarray}
P(\mathcal{F}) & \leq & \sum P\bigl(\mathcal{F}(\{j_{\Sigma}\}_{\mathcal{J}(k)})\bigr)\label{eq:union_bound}\end{eqnarray}
where the summation is over all $\{j_{\Sigma}\}_{\mathcal{J}(k)}\neq(1,\ldots,1)$.
However, a subset of indices of $\{j_{\Sigma}\}_{\mathcal{J}(k)}$
can still be equal to $1$. We expand the above summation over all
such possible subsets. Let $\mathcal{Q}_{1},\mathcal{Q}_{2},\ldots\mathcal{Q}_{N}\subseteq\mathcal{J}(k)$ satisfying (\ref{eq:cond_Q}) be such that the following holds\footnote{Again observe that it is sufficient for us to consider $\mathcal{Q}_{i}$s
which satisfy (\ref{eq:cond_Q}) due to the hierarchical structure
of the conditional codebook generation.%
}:\begin{eqnarray}
\exists i \in \{1,2,\dots,N\}:\mathcal{Q}_i \subset \mathcal{J}(k) & & \label{eq:noall_1}
\end{eqnarray}i.e., at least one of the $\mathcal{Q}_{i}$'s is a strict subset of $\mathcal{J}(k)$. Define the set:\[
A_{\mathcal{Q}_{1},\mathcal{Q}_{2},\ldots,\mathcal{Q}_{N}}=\{j_{\Sigma}\}_{\mathcal{J}(k)}\,\,\mbox{: }\forall i\begin{cases}
j_{i,\mathcal{K}}=1 & \mbox{if }\mathcal{K}\in\mathcal{Q}_{i}\\
j_{i,\mathcal{K}}\neq1 & \mbox{otherwise}\end{cases}\]
Then, we can expand (\ref{eq:union_bound}) as:\begin{equation}
P(F)\leq\sum\sum P\bigl(\mathcal{F}(\{j_{\Sigma}\}_{\mathcal{J}(k)})\bigr)\label{eq:union_bound_2}\end{equation}
where the first summation is over all $\{\mathcal{Q}_{1},\ldots,\mathcal{Q}_{N}:\mathcal{Q}_{i}\subseteq\mathcal{J}(k),$
satisfying (\ref{eq:cond_Q}) and (\ref{eq:noall_1})$\}$
and the second summation is over all $\{j_{\Sigma}\}_{\mathcal{J}(k)}\in A_{\mathcal{Q}_{1},\mathcal{Q}_{2},\ldots\mathcal{Q}_{N}}$.
We note that, due to the conditional independence of the codewords
generated, $P\bigl(\mathcal{F}(\{j_{\Sigma}\}_{\mathcal{J}(k)})\bigr)$
is the same for all $\{j_{\Sigma}\}_{\mathcal{J}(k)}\in A(\mathcal{Q}_{1},\mathcal{Q}_{2},\ldots\mathcal{Q}_{N})$,
i.e., $P\bigl(\mathcal{F}(\{j_{\Sigma}\}_{\mathcal{J}(k)})\bigr)$
depends only on $\mathcal{Q}_{1},\mathcal{Q}_{2},\ldots\mathcal{Q}_{N}$. 
We can bound $P(\mathcal{F})$ as:\begin{eqnarray}
P(\mathcal{F})\leq\sum\bigl\{|A_{\mathcal{Q}_{1},\mathcal{Q}_{2},\ldots\mathcal{Q}_{N}}|\nonumber \\
P\bigl(\mathcal{F}(\mathcal{Q}_{i};\forall i\in\Sigma)\bigr)\bigr\}\label{eq:union_bound_3}\end{eqnarray}
where $P(\mathcal{F}(\mathcal{Q}_{i};\forall i\in\Sigma))=$$P\bigl(\mathcal{F}(\mathcal{Q}_{1},\mathcal{Q}_{2},\ldots,\mathcal{Q}_{N})\bigr)$
denotes $P(\mathcal{F}(\{j_{\Sigma}\}_{\mathcal{J}(k)}))$ for some
$\{j_{\Sigma}\}_{\mathcal{J}(k)}\in A_{\mathcal{Q}_{1},\mathcal{Q}_{2},\ldots\mathcal{Q}_{N}}$
and the summation is over all $\{\mathcal{Q}_{1},\ldots,\mathcal{Q}_{N}:\mathcal{Q}_{i}\subseteq\mathcal{J}(k),\mbox{ satisfying (\ref{eq:cond_Q}) and  (\ref{eq:noall_1})} \}$.
We next bound the individual terms in the above product. Recall that
each of the bins $\mathcal{C}_{i,\mathcal{S}}(\cdot)$ have $2^{n(R_{i,\mathcal{S}}^{'}-R_{i,\mathcal{S}}^{''})}$
codewords. Using Lemma 3.1 \cite{Han_Kobayashi}, we can bound both
the terms in the above product as:\begin{eqnarray}
P\bigl(\mathcal{F}(\mathcal{Q}_{i};\forall i\in\Sigma)\bigr)\leq & \frac{2^{nH\big(\{U_{i}\}_{\mathcal{Q}_{i}^{c}}\forall i\big|\{U_{i}\}_{\mathcal{Q}_{i}}\forall i\big)}}{2^{n\sum\limits_{i\in\Sigma}\sum\limits_{\mathcal{K}\in\mathcal{Q}_{i}^{c}}H\big(U_{i,\mathcal{K}}\big|\{U_{i}\}_{\mathcal{I}_{|\mathcal{K}|+}(\mathcal{K})}\big)-4n\epsilon}}\nonumber \\
|A_{\mathcal{Q}_{1},\mathcal{Q}_{2},\ldots\mathcal{Q}_{N}}|\leq & 2^{n\sum\limits_{i\in\Sigma}\sum\limits_{\mathcal{K}\in\mathcal{Q}_{i}^{c}}(R_{i,\mathcal{K}}^{'}-R_{i,\mathcal{K}}^{''})}\label{eq:bound_terms}\end{eqnarray}
where $\mathcal{Q}_{i}^{c}=\mathcal{J}(k)-\mathcal{Q}_{i}$. Substituting
(\ref{eq:bound_terms}) in (\ref{eq:union_bound_3}), it follows that
$P(\mathcal{F})$ can be made arbitrarily small if: $\forall\mathcal{Q}_{1},\mathcal{Q}_{2},\ldots\mathcal{Q}_{N}\subseteq\mathcal{J}(k)$
satisfying (\ref{eq:cond_Q}) and (\ref{eq:noall_1}),\begin{eqnarray*}
\sum_{i\in\Sigma}\sum_{\mathcal{K}\in\mathcal{Q}_{i}^{c}}(R_{i,\mathcal{K}}^{'}-R_{i,\mathcal{K}}^{''})\leq\sum_{i\in\Sigma}\sum_{\mathcal{K}\in\mathcal{Q}_{i}^{c}}H(U_{i,\mathcal{K}}|\{U_{i}\}_{\mathcal{I}_{|\mathcal{K}|+}(\mathcal{K})})\\
-H\left(\{U_{i}\}_{\mathcal{Q}_{i}^{c}}\forall i|\{U_{i}\}_{\mathcal{Q}_{i}}\forall i\right)-4\epsilon\end{eqnarray*}
\begin{equation}
\label{eq:bound_rate_diff}\end{equation}
where $\mathcal{Q}_{i}^{c}=\mathcal{J}(k)-\mathcal{Q}_{i}$. On plugging
in the bounds for $R_{i,\mathcal{K}}^{'}$ from (\ref{eq:proof_1})
into (\ref{eq:bound_rate_diff}), we get (\ref{eq:thm_1}) in Theorem
\ref{thm:main}.

\end{proof}

\section*{Appendix B: Conditional Markov Lemma - For Mutual Covering}
\label{app:cond_mark_lemma}
\begin{figure*}
	\centering
	\subfloat[][Generalized Markov Lemma \cite{Han_Kobayashi}]{\label{fig:HK}\includegraphics[scale=0.55]{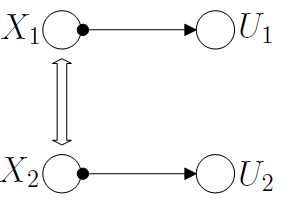}}
	\subfloat[][Conditional Markov Lemma \cite{Wagner}]{\label{fig:Wagner}\includegraphics[scale=0.55]{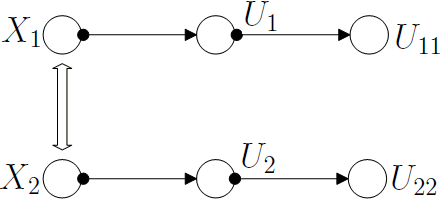}}
	\subfloat[][Conditional Markov Lemma - for Mutual Covering]{\label{fig:Cond_mark_lemma}\includegraphics[scale=0.5]{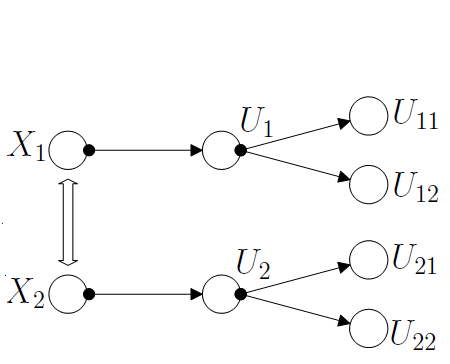}}
	\caption{Depicts the different Markov lemmas.}
\end{figure*}

It was shown in \cite{Han_Kobayashi}\footnote{We note that an earlier Markov Lemma proof appeared 
in \cite{Tung}. However the proof in \cite{Han_Kobayashi} is
easily extendible to more general settings as it is based on standard
typicality arguments. %
} that if a codeword of $U_{1}$ (denoted by $U_{1}^{*}$) is selected jointly
typical with $X_{1}^{n}$ and a codeword of $U_{2}$ (denoted by $U_{2}^{*}$)
is selected jointly typical with $X_{2}^{n}$ and if $U_{1}\leftrightarrow X_{1}\leftrightarrow X_{2}\leftrightarrow U_{2}$,
then $(U_{1}^{*},X_{1}^{n},X_{2}^{n},U_{2}^{*})$ are jointly typical.
This is called the generalized Markov lemma and is depicted in Fig.
\ref{fig:HK}. Similarly, Wagner et al. \cite{Wagner} considered the
case in which codewords of $U_{11}$ and $U_{22}$ are generated conditioned
on codewords of $U_{1}$ and $U_{2}$, respectively. They showed that
if a pair of codewords of $(U_{1},U_{11})$ (denoted by $(U_{1}^{*},U_{11}^{*})$)
are jointly typical with $X_{1}^{n}$ and a pair of codewords
of $(U_{2},U_{22})$ (denoted by $(U_{2}^{*},U_{22}^{*})$) are typical with $X_{2}^{n}$, and if $(U_{1},U_{11})\leftrightarrow X_{1}\leftrightarrow X_{2}\leftrightarrow(U_{2},U_{22})$,
then $(U_{1}^{*},U_{11}^{*},X_{1}^{n},X_{2}^{n},U_{2}^{*}.U_{22}^{*})$
are jointly typical. This is called the conditional Markov lemma for
obvious reasons and is depicted in Fig. \ref{fig:Wagner}. However, these results are not sufficient for our scenario and
we require a stronger version of the conditional Markov lemma. In
what follows, we will establish a series of lemmas, culminating with the needed variant called the conditional Markov lemma for mutual covering (Lemma  \ref{lem:Cond_Mark_Lemma_3}). 
Note that these lemmas can be easily extended to more than $2$ random
variables and layers of encoding. However, we restrict ourselves to
the $2$ variable case to keep the notation simple. We also note that
the lemmas and proofs here are applicable to more general contexts
beyond DIR.
\begin{lem}
\label{lem:Cond_Mark_Lemma_1}Let random variables $(Y,U,V_{1},V_{2})$
be given and let $y^{n}\in\mathcal{T}_{\epsilon}^{n}(Y)$. Let the subset $B_{0}(y^{n})\subset\mathcal{T}_{\epsilon}^{n}(U|y^{n})$ 
be such that:\begin{eqnarray}
2^{n(H(U|Y)-\lambda)}\leq|B_{0}(y^{n})|\leq2^{n(H(U|Y)+\lambda)}\label{eq:Cond_Mark_Lemma_B_Bound-1}\end{eqnarray}
for some $\lambda>0$. For every $u^{n}\in B_{0}(y^{n})$, let subset $B_{12}(y^{n},u^{n})\subset\mathcal{T}_{\epsilon}^{n}((V_{1},V_{2})|u^{n})$ be such that:\begin{equation}
2^{n(H(V_{1},V_{2}|U,Y)-\lambda)}\leq|B_{12}(y^{n},u^{n})|\leq2^{n(H(V_{1},V_{2}|U,Y)+\lambda)}
\label{eq:Cond_Mark_Lemma_B_Bound}\end{equation}and the following hold:\begin{eqnarray}
2^{n(H(V_{1}|U,Y)-\lambda)}\leq|B_{1}(y^{n},u^{n})|\leq2^{n(H(V_{1}|U,Y)+\lambda)}\nonumber \\
2^{n(H(V_{2}|U,Y)-\lambda)}\leq|B_{2}(y^{n},u^{n})|\leq2^{n(H(V_{2}|U,Y)+\lambda)}\nonumber \\
2^{n(H(V_{1}|U,Y,V_{2})-\lambda)}\leq|\hat{B}_{1}(y^{n},u^{n},v_{2}^{n})|\leq2^{n(H(V_{1}|U,Y,V_{2})+\lambda)}\nonumber 
\end{eqnarray}\begin{equation}
2^{n(H(V_{2}|U,Y,V_{1})-\lambda)}\leq|\hat{B}_{2}(y^{n},u^{n},v_{1}^{n})|\leq2^{n(H(V_{2}|U,Y,V_{1})+\lambda)}\label{eq:Cond_Mark_Lemma_B_Bound_2}
\end{equation}

where $\forall(v_{1}^{n},v_{2}^{n})\in B_{12}(y^{n},u^{n})$:\begin{eqnarray}
\hat{B}_{1}(y^{n},u^{n},v_{2}^{n}) & = & \{v_{1}^{n}:(v_{1}^{n},v_{2}^{n})\in B_{12}(y^{n},u^{n})\}\nonumber \\
\hat{B}_{2}(y^{n},u^{n},v_{1}^{n}) & = & \{v_{2}^{n}:(v_{1}^{n},v_{2}^{n})\in B_{12}(y^{n},u^{n})\}\nonumber \\
B_{1}(y^{n},u^{n}) & = & \{v_{1}^{n}:\exists(v_{1}^{n},v_{2}^{n})\in B_{12}(y^{n},u^{n})\}\nonumber \\
B_{2}(y^{n},u^{n}) & = & \{v_{2}^{n}:\exists(v_{1}^{n},v_{2}^{n})\in B_{12}(y^{n},u^{n})\}\end{eqnarray}
Let $R_{0},R_{1}$ and $R_{2}$ be given positive rates. Let $\overline{U}_{j}(j=1,\ldots,2^{nR_{0}})$ be random
variables drawn independently and uniformly from $\mathcal{T}_{\epsilon}^{n}(U)$.
For each $\overline{U}_{j}$, let $\overline{V}_{jk}^{1}(k=1,\ldots,2^{nR_{1}})$
and $\overline{V}_{jk}^{2}(k=1,\ldots,2^{nR_{2}})$ be random variables
drawn independently and uniformly from $\mathcal{T}_{\epsilon}^{n}(V_{1}|\bar{U}_{j})$
and $\mathcal{T}_{\epsilon}^{n}(V_{2}|\bar{U}_{j})$, respectively.
Then for $n$ sufficiently large, \begin{align}
P\left(\nexists j,k_{1},k_{2}:\overline{U}_{j}\in B_{0}(y^{n}),\,(\bar{V}_{jk_{1}}^{1},\bar{V}_{jk_{2}}^{2})\in B_{12}(y^{n},\overline{U}_{j})\right)\nonumber \\
\leq\delta(\epsilon) \label{eq:Cond_Mark_Lemma_Prob_Cond1}\end{align}
 where $\delta(\epsilon)\rightarrow0$ as $\epsilon\rightarrow0$,
if the rates $R_{0}$,$R_{1}$ and $R_{2}$ satisfy: \begin{eqnarray}
R_{0} & \geq & I(Y;U)+7\lambda+19\epsilon\nonumber \\
R_{0}+R_{1} & \geq & I(Y;V_{1},U)+8\lambda+17\epsilon\nonumber \\
R_{0}+R_{2} & \geq & I(Y;V_{2},U)+8\lambda+17\epsilon\nonumber \\
R_{0}+R_{1}+R_{2} & \geq & I(Y;V_{1},V_{2},U)+I(V_{1};V_{2}|U)\nonumber \\
 &  & +6\lambda+15\epsilon\label{eq:Cond_Mark_Lemma_Rate_Cond}\end{eqnarray}

\end{lem}
\begin{proof} Define the random variable $\mathcal{X}_{j,k_{1},k_{2}}$ as :\begin{equation}
\mathcal{X}_{j,k_{1},k_{2}}=\begin{cases}
1 & \mbox{if }\bar{U}_{j}\in B_{0}(y^{n}),(\bar{V}_{k_{1}}^{1},\bar{V}_{k_{2}}^{2})\in B_{12}(y^{n},\bar{U}_{j})\\
0 & \mbox{else}\end{cases}\end{equation}
 Denote by $\mathcal{X}=\sum_{j,k_{1},k_{2}}\mathcal{X}_{j,k_{1},k_{2}}$. Observe that the probability in (\ref{eq:Cond_Mark_Lemma_Prob_Cond1}) is equal to $P(\mathcal{X}=0)$. From Chebychev's inequality, we have: \begin{eqnarray}
P\left(\mathcal{X}=0\right)\leq\frac{4Var(\mathcal{X})}{(E[\mathcal{X}])^{2}}\label{eq:Cond_Mark_Cheb}\end{eqnarray}
 Next we have the following from (\ref{eq:Cond_Mark_Lemma_B_Bound-1}) and
(\ref{eq:Cond_Mark_Lemma_B_Bound}):\begin{eqnarray}
E[\mathcal{X}] & = & \sum_{j,k_{1},k_{2}}E[\mathcal{X}_{j,k_{1},k_{2}}]\nonumber \\
 & =^{(a)} & 2^{n(R_{0}+R_{1}+R_{2})}P(\mathcal{X}_{1,1,1})\nonumber \\
 & \geq & 2^{n(R_{0}+R_{1}+R_{2})}\frac{2^{n(H(V_{1},V_{2},U|Y))-2\lambda-5\epsilon)}}{2^{n(H(U)+H(V_{1}|U)+H(V_{2}|U))}}\nonumber \\
\label{eq:Cond_Mark_Ex}\end{eqnarray}
where equality in (a) holds because the random variables $\overline{U}_{j}$, $\overline{V}_{jk}^{1}$ 
and $\overline{V}_{jk}^{2}$ are drawn independently and uniformly from their respective typical sets. Also, using (\ref{eq:Cond_Mark_Lemma_B_Bound}) and (\ref{eq:Cond_Mark_Lemma_B_Bound_2}),
we can bound $E[\mathcal{X}^{2}]$ as:\begin{eqnarray}
E[\mathcal{X}^{2}] & = & \sum_{j^{1},k_{1}^{1},k_{2}^{1}}\sum_{j^{2},k_{1}^{2},k_{2}^{2}}E[\mathcal{X}_{j^{1},k_{1}^{1},k_{2}^{1}}\mathcal{X}_{j^{2},k_{1}^{2},k_{2}^{2}}]\nonumber \\
 & \leq & 2^{n(R_{0}+R_{1}+R_{2})}P_{1}+2^{n(R_{0}+2R_{1}+2R_{2})}P_{2}\nonumber \\
 &  & +2^{n(R_{0}+2R_{1}+R_{2})}P_{3}+2^{n(R_{0}+R_{1}+2R_{2})}P_{4}\nonumber \\
 &  & +2^{2n(R_{0}+R_{1}+R_{2})}P_{1}^{2}\label{eq:Cond_Mark_Ex2}\end{eqnarray}
where \begin{eqnarray}
P_{1} & = & \frac{2^{n(H(V_{1},V_{2},U|Y)+2\lambda+5\epsilon)}}{2^{n(H(U)+H(V_{1}|U)+H(V_{2}|U))}}\nonumber \\
P_{2} & = & \frac{2^{n(H(U|Y)+2H(V_{1},V_{2}|U,Y)+3\lambda+9\epsilon)}}{2^{n(H(U)+2H(V_{1}|U)+2H(V_{2}|U))}}\nonumber \\
P_{3} & = & \frac{2^{n(H(V_{1},U|Y)+2H(V_{2}|U,Y,V_{1})+4\lambda+7\epsilon)}}{2^{n(H(V_{1},U)+2H(V_{2}|U))}}\nonumber \\
P_{4} & = & \frac{2^{n(2H(V_{1}|U,Y,V_{2})+H(V_{2},U|Y))+4\lambda+7\epsilon)}}{2^{n(2H(V_{1}|U)+H(V_{2},U))}}\label{eq:Cond_Mark_p1p2p3}\end{eqnarray}
On substituting (\ref{eq:Cond_Mark_Ex}),(\ref{eq:Cond_Mark_Ex2})
and (\ref{eq:Cond_Mark_p1p2p3}) in (\ref{eq:Cond_Mark_Cheb}), we
have:\begin{eqnarray}
 &  & P\left(\mathcal{X}=0\right)\leq\nonumber \\
 &  & 4\Bigl[\delta(\epsilon)+2^{-n(R_{0}-I(Y;U)-7\lambda-19\epsilon)}\nonumber \\
 &  & +2^{-n(R_{0}+R_{1}-I(Y;V_{1},U)-8\lambda-17\epsilon)}\nonumber \\
 &  & +2^{-n(R_{0}+R_{2}-I(Y;V_{2},U)-8\lambda-17\epsilon)}\nonumber \\
 &  & +2^{-n(R_{0}+R_{1}+R_{2}-I(Y;V_{1},V_{2},U)-I(V_{1};V_{2}|U)-6\lambda-15\epsilon)}\Bigr]\end{eqnarray}
which can be made arbitrarily small if the rates satisfy (\ref{eq:Cond_Mark_Lemma_Rate_Cond}).
\end{proof}
\begin{lem}
\label{lem:Cond_Mark_Lemma_2}Let $W,Y,U,V_{1}$ and $V_{2}$ be random
variables with values in finite sets $\mathcal{W},\mathcal{Y},\mathcal{U},\mathcal{V}_{1}$
and $\mathcal{V}_{2}$, respectively. Let $W^{*}$ be a random variable
with values in $\mathcal{W}^{n}$, such that:\begin{equation}
W^{*}\leftrightarrow Y^{n}\leftrightarrow(U^{n},V_{1}^{n},V_{2}^{n})\label{eq:Cond_Mark_Lemma_Markov Cond}\end{equation}
Let $R_{0}$, $R_{1}$ and $R_{2}$ be given positive rates. Let $\overline{U}_{i}|_{i=1}^{2^{nR_{0}}}$ denote independent
random variables chosen uniformly with replacement from $\mathcal{T}_{\epsilon}^{n}(U)$.
Let $\overline{V}_{i,j}^{1}(i=1,\ldots,2^{R_{0}},\,\, j=1,\ldots,2^{nR_{1}})$
and $\overline{V}_{i,j}^{2}(i=1,\ldots,2^{R_{0}},\,\, j=1\ldots2^{nR_{2}})$
be random variables drawn independently and uniformly from $\mathcal{T}_{\epsilon}^{n}(V_{1}|\overline{U}_{i})$
and $\mathcal{T}_{\epsilon}^{n}(V_{2}|\overline{U}_{i})$, respectively
$\forall i$. Further, let, \begin{equation}
P(W^{*},Y^{n},U^{n},V_{1}^{n},V_{2}^{n}\in\mathcal{T}_{\epsilon}^{n}(W,Y,U,V_{1},V_{2}))\geq1-\eta\label{eq:Cond_Mark_Lemma_Given}\end{equation}
Also, suppose $\forall v_{1}^{n}\in\mathcal{T}_{\epsilon}^{n}(V_{1})$
and $v_{2}^{n}\in\mathcal{T}_{\epsilon}^{n}(V_{2})$:\begin{eqnarray}
P\left((W^{*},Y^{n},U^{n},V_{1}^{n})\in\mathcal{T}_{\epsilon}^{n}\bigl|V_{2}^{n}=v_{2}^{n}\right)\geq1-\eta\nonumber \\
P\left((W^{*},Y^{n},U^{n},V_{2}^{n})\in\mathcal{T}_{\epsilon}^{n}\bigl|V_{1}^{n}=v_{1}^{n}\right)\geq1-\eta\label{eq:Cond_Mark_Lemma_Given_2-1}\end{eqnarray}
Then for $n$ sufficiently large, there exists functions $U^{*}(y^{n})$,
$V_{1}^{*}(y^n,U^{*})$ and $V_{2}^{*}(y^n,U^{*})$, such that:

i) \textup{$\overline{U}^{*}(y^{n})=\overline{U}_{i}$ (for some $i\in\{1,\ldots,2^{R_{0}}\}$)$\Rightarrow V_{1}^{*}(y^{n},\overline{U}^{*})=\bar{V}_{i,j_{1}}^{1}$,
$V_{2}^{*}(y^{n},\overline{U}^{*})=\bar{V}_{i,j_{2}}^{2}$ for some
$j_{1}\in\{1,\ldots,2^{nR_{1}}\}$ and $j_{2}\in\{1,\ldots,2^{nR_{2}}\}$}

%(W,Y,U,V_{1},V_{2}))
%(W,Y,U,V_{1}|V_{2}^{*})
%(W,Y,U,V_{1}|V_{1}^{*})
ii) \begin{eqnarray}
P((W^{*},Y^{n},\bar{U}^{*},V_{1}^{*},V_{2}^{*})\in\mathcal{T}_{\epsilon} \geq1-\delta(\epsilon)& &\nonumber \\
P((W^{*},Y^{n},\bar{U}^{*},V_{1}^{*})\in\mathcal{T}_{\epsilon}\bigl|V_{2}^{*})\geq1-\delta(\epsilon)& &  \nonumber \\
P((W^{*},Y^{n},\bar{U}^{*},V_{2}^{*})\in\mathcal{T}_{\epsilon}\bigl|V_{1}^{*})\geq1-\delta(\epsilon)& &\label{eq:Cond_Mark_Lemma_Them_Eq}\end{eqnarray}
 for some $\delta(\epsilon)\rightarrow0$ as $\epsilon\rightarrow0$,
if the rates $R_{0}$,$R_{1}$ and $R_{2}$ satisfy: \begin{eqnarray}
R_{0} & \geq & I(Y;U)+40\epsilon\nonumber \\
R_{0}+R_{1} & \geq & I(Y;V_{1},U)+41\epsilon\nonumber \\
R_{0}+R_{2} & \geq & I(Y;V_{2},U)+41\epsilon\label{eq:Cond_Mark_Lemma_Thm_Rate_cond} \\
R_{0}+R_{1}+R_{2} & \geq & I(Y;V_{1},V_{2},U)+I(V_{1};V_{2}|U)+33\epsilon\nonumber\end{eqnarray}
 \end{lem}
\begin{proof}
Let us expand (\ref{eq:Cond_Mark_Lemma_Given}) as:\begin{eqnarray}
\sum_{y^{n}\in\mathcal{Y}^{n}}\Bigl\{ P\Bigl((W^{*},U^{n},V_{1}^{n},V_{2}^{n})\in\mathcal{T}_{\epsilon}^{n}\Bigl|Y^{n}=y^{n}\Bigr)\nonumber \\
P(Y^{n}=y^{n})\Bigr\}\geq1-\eta\nonumber \\
\label{eq:Cond_Mark_Lemma_Proof1}\end{eqnarray}
 Let,\begin{eqnarray*}
A\triangleq\Bigl\{ y^{n}:P\Bigl((W^{*},U^{n},V_{1}^{n},V_{2}^{n})\in\mathcal{T}_{\epsilon}^{n}\Bigl|Y^{n}=y^{n}\Bigr)\\
\geq1-\sqrt{\eta}\Bigr\}\end{eqnarray*}
 and \begin{equation}
A_{0} \triangleq A\bigcap\mathcal{T}_{\epsilon}^{n}(Y)\end{equation}
Then using the reverse Markov inequality, we can show that (similar to
\cite{Han_Kobayashi,Wagner}):\begin{equation}
P(Y^{n}\in A_{0})\geq1-\delta_{1}\label{eq:Cond_Mark_Lemma_PA0}\end{equation}
where $\delta_{1}=\sqrt{\eta}+\epsilon$. Then for any $y^{n}\in A_{0}$,
we have:\begin{eqnarray}
\sum_{u^{n}}\Bigl\{ P\Bigl((W^{*},V_{1}^{n},V_{2}^{n})\in\mathcal{T}_{\epsilon}^{n}\Bigl|Y^{n}=y^{n},U^{n}=u^{n}\Bigr)\nonumber \\
P\Bigl(U^{n}=u^{n}\Bigl|Y^{n}=y^{n}\Bigr)\Bigr\}\geq1-\sqrt{\eta}\end{eqnarray}
Let,\begin{eqnarray*}
B(y^{n})\triangleq\Bigl\{ u^{n}:P\Bigl((W^{*},V_{1}^{n},V_{2}^{n})\in\mathcal{T}_{\epsilon}^{n}\Bigl|Y^{n}=y^{n},U^{n}=u^{n}\Bigr)\\
\geq1-\sqrt[4]{\eta}\Bigr\}\end{eqnarray*}
 \begin{equation}
B_{0} \triangleq B\bigcap\mathcal{T}_{\epsilon}^{n}(U|y^{n})\end{equation}
Using the reverse Markov inequality, we again have:\begin{equation}
P\Bigl(U^{n}\in B_{0}(y^{n})\Bigl|Y^{n}=y^{n}\Bigr)\geq1-\delta_{2}\label{eq:Cond_Mark_Lemma_PB0}\end{equation}
where $\delta_{2}=\sqrt[4]{\eta}+\epsilon$. Hence for any $y^{n}\in A_{0}$
and $u^{n}\in B_{0}(y^{n})$ we have: \begin{eqnarray}
\sum_{u^{n},v_{1}^{n},v_{2}^{n}}P\Bigl(V_{1}^{n}=v_{1}^{n},V_{2}^{n}=v_{2}^{n}\Bigl|Y^{n}=y^{n},U^{n}=u^{n}\Bigr)\nonumber \\
Q(y^{n},u^{n},v_{1}^{n},v_{2}^{n})\geq1-\sqrt[4]{\eta}\end{eqnarray}
where we denote by $Q(\mathcal{S})$$=P\Bigl(W^{*}\in\mathcal{T}_{\epsilon}^{n}(W|\mathcal{S})\Bigl|Y^{n}=y^{n}\Bigr)$
for any set of sequences $\mathcal{S}$. Note that we have used the Markov condition (\ref{eq:Cond_Mark_Lemma_Markov Cond})
in the above equation. Now define sets $\tilde{B}_{12}(y^{n},u^{n})$
and $B_{12}(y^{n},u^{n})$ for any $y^{n}\in A_{0}$ and $u^{n}\in B_{0}(y^{n})$
such that:

\begin{eqnarray*}
\tilde{B}_{12}(y^{n},u^{n})\triangleq\Bigl\{(v_{1}^{n},v_{2}^{n}):Q(y^{n},u^{n},v_{1}^{n},v_{2}^{n})\geq1-\sqrt[8]{\eta}\Bigr\}\end{eqnarray*}
 \begin{equation}
B_{12}(y^{n})\triangleq\tilde{B}_{12}(y^{n})\bigcap\mathcal{T}_{\epsilon}^{n}(U,V_{1},V_{2}|y^{n})\label{eq:Cond_Mark_Lemma_B12_Defn}\end{equation}
Then using the reverse Markov inequality, we can show that:\begin{equation}
P\Bigl((V_{1}^{n},V_{2}^{n})\in B_{12}(y^{n})\Bigl|Y^{n}=y^{n},U^{n}=u^{n}\Bigr)\geq1-\delta_{3}\label{eq:Cond_Mark_Lemma_PB12}\end{equation}
where $\delta_{3}=\sqrt[8]{\eta}+\epsilon$. Then from (\ref{eq:Cond_Mark_Lemma_PB0}),
(\ref{eq:Cond_Mark_Lemma_PB12}) and Lemma 3.1(f) in \cite{Han_Kobayashi},
for $n$ sufficiently large, we have: \begin{align}
2^{n(H(U|Y)-3\epsilon)} & \leq|B_{0}(y^{n})|\leq2^{n(H(U|Y)+\epsilon)}\nonumber \\
2^{n(H(V_{1},V_{2}|Y,U)-3\epsilon)} & \leq|B_{12}(y^{n},u^{n})|\leq2^{n(H(V_{1},V_{2}|Y,U)+\epsilon)}\label{eq:Cond_Mark_Lemma_size_B0B12}\end{align}
Note that we have two of the sets required by Lemma \ref{lem:Cond_Mark_Lemma_1}. However,
we further require bounds on the projections of $B_{12}(y^{n},u^{n})$ (as
in (\ref{eq:Cond_Mark_Lemma_B_Bound_2})) to invoke Lemma \ref{lem:Cond_Mark_Lemma_1}.
Towards obtaining these bounds, we note that the following inequalities can be
shown directly from (\ref{eq:Cond_Mark_Lemma_Given}):\begin{eqnarray}
P\left((W^{*},Y^{n},U^{n},V_{1}^{n})\in\mathcal{T}_{\epsilon}^{n}\right)\geq1-\eta\nonumber \\
P\left((W^{*},Y^{n},U^{n},V_{2}^{n})\in\mathcal{T}_{\epsilon}^{n}\right)\geq1-\eta\label{eq:Cond_Mark_Lemma_derived}\end{eqnarray}
Expanding (\ref{eq:Cond_Mark_Lemma_derived}) instead of (\ref{eq:Cond_Mark_Lemma_Given}) and repeating all steps from (\ref{eq:Cond_Mark_Lemma_Proof1}) through (\ref{eq:Cond_Mark_Lemma_size_B0B12}), we obtain:\begin{align}
2^{n(H(V_{1}|Y,U)-3\epsilon)} & \leq|B_{1}(y^{n},u^{n})|\leq2^{n(H(V_{1}|Y,U)+\epsilon)}\nonumber \\
2^{n(H(V_{2}|Y,U)-3\epsilon)} & \leq|B_{2}(y^{n},u^{n})|\leq2^{n(H(V_{2}|Y,U)+\epsilon)}\end{align}
where \begin{align}
B_{1}(y^{n},u^{n}) & =\{v_{1}^{n}:\exists(v_{1}^{n},v_{2}^{n})\in B_{12}(y^{n},u^{n})\}\nonumber \\
B_{2}(y^{n},u^{n}) & =\{v_{2}^{n}:\exists(v_{1}^{n},v_{2}^{n})\in B_{12}(y^{n},u^{n})\}\end{align}
Similarly, it is easy to show that expanding (\ref{eq:Cond_Mark_Lemma_Given_2-1}) instead
of (\ref{eq:Cond_Mark_Lemma_Given}) leads to:\begin{align}
2^{n(H(V_{1}|Y,U,V_2)-3\epsilon)} & \leq|\hat{B}_{1}(y^{n},u^{n},v_{2}^{n})|\leq2^{n(H(V_{1}|Y,U,V_2)+\epsilon)}\nonumber \\
2^{n(H(V_{2}|Y,U,V_{1})-3\epsilon)} & \leq|\hat{B}_{2}(y^{n},u^{n},v_{1}^{n})|\leq2^{n(H(V_{2}|Y,U,V_1)+\epsilon)}\end{align}
where $\forall v_{1}^{n}\in B_{1}(y^{n},u^{n})$ and $v_{2}^{n}\in B_{2}(y^{n},u^{n})$,
\begin{align}
\hat{B}_{1}(y^{n},u^{n},v_{2}^{n}) & =\{v_{1}^{n}:(v_{1}^{n},v_{2}^{n})\in B_{12}(y^{n},u^{n})\}\nonumber \\
\hat{B}_{2}(y^{n},u^{n},v_{1}^{n}) & =\{v_{2}^{n}:(v_{1}^{n},v_{2}^{n})\in B_{12}(y^{n},u^{n})\}\end{align}
We now have sets $B_{0}$ and $B_{12}$ satisfying all the bounds
 as required in Lemma \ref{lem:Cond_Mark_Lemma_1}. Hence, we can define the functions
$U^{*},V_{1}^{*}$ and $V_{2}^{*}$ as follows. $U^{*}(y^{n})=\overline{U}_{i}$
if $\bar{U}_{i}\in B_{0}(y^{n})$. If no such $\bar{U}_{i}$ exists,
we set $U^{*}(y^{n})=\bar{U}_{1}$. Next, if there exists a pair $(\overline{V}_{i,j_{1}}^{1},\overline{V}_{i,j_{2}}^{2})$
such that $(\overline{V}_{i,j_{1}}^{1},\overline{V}_{i,j_{2}}^{2})\in B_{12}(y^{n},\overline{U}_{i})$,
then define $(V_{1}^{*}(y^{n},U^{*}),V_{2}^{*}(y^{n},U^{*}))=(\overline{V}_{i,j_{1}}^{1},\overline{V}_{i,j_{2}}^{2})$.
If there exists no such pair, define $(V_{1}^{*}(y^{n},U^{*}),V_{2}^{*}(y^{n},U^{*}))=(\overline{V}_{i,1}^{1},\overline{V}_{i,1}^{2})$.

It follows from the rate conditions in (\ref{eq:Cond_Mark_Lemma_Thm_Rate_cond}),
Lemma \ref{lem:Cond_Mark_Lemma_1} with $\lambda=3\epsilon$ and the
bounds on set sizes that:\begin{align}
P\left(U^{*}\in B_{0}(Y^{n}),(V_{1}^{*},V_{2}^{*})\in B_{12}(Y^{n},U^{*})\Bigl|Y^{n}\in A_{0}\right)\nonumber \\
\geq1-\delta(\epsilon)\end{align}
for some $\delta(\epsilon)\rightarrow0$ as $\epsilon\rightarrow0$.
Note that $y^{n}\in A_{0}$, $U^{*}\in B_{0}(Y^{n})$ and $(V_{1}^{*},V_{2}^{*})\in B_{12}(Y^{n},U^{*})$
imply that $(y^{n},U^{*},V_{1}^{*},V_{2}^{*})\in\mathcal{T}_{\epsilon}^{n}(Y,U,V_{1},V_{2})$.
We then have,\begin{equation}
P(W^{*},Y^{n},U^{*},V_{1}^{*},V_{2}^{*}\in T_{\epsilon}^{n})\geq P(E_{1})P(E_{2}|E_{1}) \label{eq:bound_p_1}\end{equation}
where events $E_{1}$ and $E_{2}$ are defined as: \begin{eqnarray}
E_{1} & = & \{Y^{n}\in A_{0},U^{*}\in B_{0},(V_{1}^{*},V_{2}^{*})\in B_{12}\}\nonumber \\
E_{2} & = & \{W^{*}\in T_{\epsilon}^{n}(W|Y^{n},U^{*},V_{1}^{*},V_{2}^{*})\}\end{eqnarray}
From (\ref{eq:Cond_Mark_Lemma_PA0}), (\ref{eq:Cond_Mark_Lemma_PB12})
and (\ref{eq:Cond_Mark_Lemma_B12_Defn}), we obtain bounds on $P(E_{1})$
and $P(E_{2}|E_{1})$:\begin{eqnarray}
P(E_{1}) & \geq & 1-\delta_{1}-\delta_{2}-\delta_{3}\nonumber \\
P(E_{2}|E_{1}) & \geq & 1-\sqrt[8]{\eta}\end{eqnarray}
On substituting in (\ref{eq:bound_p_1}), we obtain the first bound in (\ref{eq:Cond_Mark_Lemma_Them_Eq}). The
other two bounds in (\ref{eq:Cond_Mark_Lemma_Them_Eq}) can be shown using similar arguments.\end{proof}
\begin{lem}
\label{lem:Cond_Mark_Lemma_3}{\bf Conditional Markov Lemma - for Mutual Covering}: Suppose that $(X_{1},X_{2},U_{1},U_{2},U_{11},U_{12},U_{21},U_{22})$
are random variables taking values in arbitrary finite sets $(\mathcal{X}_{1},\mathcal{X}_{2},\mathcal{U}_{1},\mathcal{U}_{2},\mathcal{U}_{11},\mathcal{U}_{12},\mathcal{U}_{21},\mathcal{U}_{22})$, 
respectively. Let the random variables satisfy the following Markov
condition:\begin{equation}
(U_{1},U_{11},U_{12})\leftrightarrow X_{1}\leftrightarrow X_{2}\leftrightarrow(U_{2},U_{21},U_{22})\label{eq:LL_Mark_Cond}\end{equation}
 Let $\overline{U}_{1,i}:i=1,\ldots,2^{nR_{1}}$ and $\overline{U}_{2,i}:i=1,\ldots,2^{nR_{2}}$
be independent codewords of length $n$ each generated using the marginals
$P(U_{1})$ and $P(U_{2})$, respectively. Let $2^{nR_{11}}$ and $2^{nR_{12}}$
codewords of $U_{11}$ and $U_{12}$ (denoted by $\overline{U}_{11,ij}$
and $\overline{U}_{12,ij}$), respectively, be generated conditioned
on each codeword $\overline{U}_{1,i}$. Similarly generate codewords
of $U_{21}$ and $U_{22}$ at rates $R_{21}$ and $R_{22}$, respectively, 
conditioned on the codewords of $U_{2}$. Then for $n$ sufficiently
large, there exists functions $U_{1}^{*}(X_{1}^{n})$,$U_{2}^{*}(X_{2}^{n})$,
$U_{11}^{*}(X_{1}^{n},U_{1}^{*})$,$U_{12}^{*}(X_{1}^{n},U_{1}^{*})$,$U_{21}^{*}(X_{2}^{n},U_{2}^{*})$
and $U_{22}^{*}(X_{2}^{n},U_{2}^{*})$ taking values in $\mathcal{U}_{1}^{n},\mathcal{U}_{2}^{n},\mathcal{U}_{11}^{n},\mathcal{U}_{12}^{n},\mathcal{U}_{21}^{n}$
and $\mathcal{U}_{22}^{n}$, respectively, such that:\begin{equation}
P\left((X_{1}^{n},X_{2}^{n},U_{1}^{*},U_{2}^{*},U_{11}^{*},U_{12}^{*},U_{21}^{*},U_{22}^{*})\in\mathcal{T}_{\epsilon}^{n}\right)\geq1-\delta(\epsilon)\label{eq:LL_TP}\end{equation}
 where $\delta(\epsilon)\rightarrow0$ as $\epsilon\rightarrow0$
if the rates satisfy:\begin{eqnarray}
  R_{1}&>&I(X_{1};U_{1}),\nonumber\\ R_{2}&>&I(X_{2};U_{2}) \nonumber\\
  R_{1}+R_{11}&>&I(X_{1};U_{11},U_{1}),\nonumber \\ R_{1}+R_{12}&>&I(X_{1};U_{12},U_{1}),\nonumber \\
  R_{2}+R_{21}&>&I(X_{2};U_{21},U_{2}),\nonumber \\ R_{2}+R_{22}&>&I(X_{2};U_{22},U_{2}),\nonumber \\
  R_{1}+R_{11}+R_{12}&>&I(X_{1};U_{11}U_{12},U_{1})\nonumber \\ & &+I(U_{11};U_{12}|U_{1}),\nonumber \\
  R_{2}+R_{22}+R_{21}&>&I(X_{2};U_{21},U_{22},U_{2})\nonumber \\ & &+I(U_{21};U_{22}|U_{2})\label{eq:LL_Rate_Constaint}\end{eqnarray}

\end{lem}
Note that this lemma can be easily extended to the more general case
of multiple random variables and multiple layers of encoding using induction (see \cite{Han_Kobayashi} for the general
methodology). While we use the more general version in the proof of Theorem
\ref{thm:main} in Appendix \ref{app:Proof-of-Theorem}, we restrict 
to the simpler case here for ease of understanding and to avoid
complex notation.
\begin{proof}
We note that from standard arguments \cite{EGC,Gamal_notes,VKG},
it follows that if the rates satisfy (\ref{eq:LL_Rate_Constaint}),
then there exists functions $U_{1}^{*}(X_{1}^{n})$, $U_{11}^{*}(X_{1}^{n},U_{1}^{*})$
and $U_{12}^{*}(X_{1}^{n},U_{1}^{*})$ such that: \begin{equation}
P\left((X_{1}^{n},U_{1}^{*},U_{11}^{*},U_{12}^{*})\in\mathcal{T}_{\epsilon}^{n}\right)\geq1-\delta(\epsilon)\label{eq:LL_MD_eq}\end{equation}
 for some $\delta(\epsilon)\rightarrow0$ as $\epsilon\rightarrow0$.
Also, note that $X_{2}^{n}$ is drawn according to the right conditional
PMF given $X_{1}^{n}$. Hence, we have:\begin{equation}
P\left((X_{1}^{n},X_{2}^{n},U_{1}^{*},U_{11}^{*},U_{12}^{*})\in\mathcal{T}_{\epsilon}^{n}\right)\geq1-\delta(\epsilon)\label{eq:LL_MD_eq2}\end{equation}
 What remains for us to show is that there exists functions $U_{2}^{*}(X_{2}^{n})$,
$U_{21}^{*}(X_{2}^{n},U_{2}^{*})$ and $U_{22}^{*}(X_{2}^{n},U_{2}^{*})$,
taking values in $\mathcal{U}_{2}^n,\mathcal{U}_{21}^n,\mathcal{U}_{22}^n$,
jointly typical with $X_{1}^{n},X_{2}^{n},U_{1}^{*},U_{11}^{*},U_{12}^{*}$.
We invoke Lemma \ref{lem:Cond_Mark_Lemma_2} with $W^{*}=(X_{1}^{n},U_{1}^{*},U_{11}^{*},U_{12}^{*})$,
$Y^{n}=X_{2}^{n}$, $U=U_{2}$, $V_{1}=U_{21}$ and $V_{2}=U_{22}$.
Note that given (\ref{eq:LL_Rate_Constaint}) and (\ref{eq:LL_MD_eq2}),
conditions (\ref{eq:Cond_Mark_Lemma_Markov Cond}),(\ref{eq:Cond_Mark_Lemma_Given})
and (\ref{eq:Cond_Mark_Lemma_Given_2-1}) are satisfied (for a formal
proof of this claim, refer to \cite{Gamal_notes}). Hence, it follows
from Lemma \ref{lem:Cond_Mark_Lemma_2} that there exist functions $U_{2}^{*}(X_{2}^{n})$,  
$U_{21}^{*}(X_{2}^{n},U_{2}^{*})$ and $U_{22}^{*}(X_{2}^{n},U_{2}^{*})$
such that:\begin{equation}
P\left((X_{1}^{n},X_{2}^{n},U_{1}^{*},U_{2}^{*},U_{11}^{*},U_{12}^{*},U_{21}^{*},U_{22}^{*})\in\mathcal{T}_{\epsilon}^{n}\right)\geq1-\delta(\epsilon)\label{eq:LL_Proved}\end{equation}
 thus proving the lemma.
\end{proof}

\end{document}